\documentclass[opre]{informs3nothing}  

\DoubleSpacedXI 

\usepackage{tikz}
\usepackage{gensymb}

\usepackage{natbib}
 \bibpunct[, ]{(}{)}{,}{a}{}{,}%
 \def\bibfont{\small}%
\graphicspath{{./figures/}}
\newcommand{\vx}{\mathbf{x}}
 \newcommand{\vz}{\mathbf{z}}
 \renewcommand{\Pr}{\mathbb{P}}
 
 \DeclareMathOperator*{\maximize}{maximize}
 \DeclareMathOperator*{\minimize}{minimize}

\TheoremsNumberedThrough     
\ECRepeatTheorems

\EquationsNumberedThrough    

\TITLE{Building a Location-Based Set of Social Media Users}
\RUNTITLE{Building a Location-Based Set of Social Media Users}

\RUNAUTHOR{Marks and Zaman}
\ARTICLEAUTHORS{%
\AUTHOR{Christopher Marks}
\AFF{%
Operations Research Center, Massachusetts Institute of Technology, 
Charles Stark Draper Laboratory, 
555 Technology Square, 
Cambridge, MA~  02139, 
\EMAIL{cemarks@mit.edu}
}
\AUTHOR{Tauhid Zaman}
\AFF{%
Department of Operations Management, Sloan School of Management, Massachusetts
Institute of Technology,  
77 Massachusetts Ave., 
Cambridge, MA~ 02139,  
\EMAIL{zlisto@mit.edu}
}
}

\ABSTRACT{In many instances one may want to gain situational awareness in an environment
	by monitoring the content of local social media users.  Often the challenge is how
	to build a set of users from a target location.  Here we introduce a method for building such a set of users by using an \emph{expand-classify} approach which begins with a small set of seed users from the target location
and then iteratively collects their neighbors and then classifies their  locations.
We perform this classification using maximum likelihood estimation on a factor graph model which incorporates features of the user profile and also
social network connections.  We show that maximum likelihood estimation
reduces to solving a minimum cut problem on an appropriately defined graph.  
We are able to obtain several thousand users within a few hours for many diverse locations using our approach.   
Using geo-located data, we find that our approach typically achieves good accuracy for population centers with less than 500,000 inhabitants, 
while for larger cities performance degrades somewhat.  We also find that our approach is able to collect many more users with higher accuracy 
than existing search methods. Finally, we show that by studying the content of location specific users obtained with our approach, we can identify the onset of significant social unrest in locations such as the Philippines. 
}

\SUBJECTCLASS{  
    Social networks, social media, geo-location, situational awareness, factor graph, 
}

\begin{document}

\maketitle

%


\section{Introduction}
    In many situations one wants to obtain situational awareness in an environment.   Consider a situation in which a     
    dangerous disease is reported in a remote location (e.g., see \cite{usnews-casimiro}).  For a multinational, nongovernmental, or governmental organization, monitoring such an outbreak could be difficult and costly.  If an agency had lists of all of social media users in the potentially affected areas, observing the content people are posting could be useful in ascertaining whether the disease is spreading.  Similar to this situations is the emergency response to disaster events or events related to political or social unrest.    If authorities knew which social media users were located in the affected area, they could monitor their social media content before, during, and after the disaster event.  This could serve as a powerful aid in gaining situational awareness, assessing the event's impact, and coordinating an effective response \citep{yin2012using,gao2011harnessing,sutton2008backchannels,yates2011emergency,merchant2011integrating}.
    
    Monitoring geo-located social media users has applications beyond these types of emergency response situations.  Consider someone looking to open a local business.  In order to learn about her prospective customer base and the local business climate, she might conduct Internet searches, read local media publications, and even lookup public data sources such as tax and census data.  What if, in addition to these sources of information, she had a set of all of the Twitter users in the local area, and therefore had access to all of their publicly available tweets?  Then she could mine their content to identify potential customers who she could contact.
    
    In addition to business applications, this method might be useful to political campaigns, which often conduct polls and surveys  to better understand what issues are important to people in certain local areas.  Observing the social media content of all of the users in the targeted locations could provide a deeper understanding of local political sentiment which may not be evident in national polls or surveys.  This could enhance micro-targeting campaign efforts.

\subsection{Expand-Classify Approach}
All of these different applications show the importance of being able to obtain a set of location specific social media users.  However, today it is not very easy to obtain such a set.  If one wanted to collect users in a location, there are two approaches that are available.  One is to search a social network for any content which is geo-tagged in the location of interest.  The second is to search the social network for any content or user profiles with keywords of interest, such as the name of the location.  Relying upon geo-tagged content will produce limited results because very often content does not have such information.  For instance, in the social media site Twitter it has been found that less than 1\% of the content is geo-tagged \citep{hecht2011tweets}.  Keyword based search may face a similar problem as many times users will not directly refer to their location in their content or profile.  Therefore, there are potentially many users in a location that would not be found by either of these methods.  

To overcome the existing limitations, we develop an {\it expand-classify} approach to iteratively collect a set of location specific social media users, with an emphasis on small or medium-sized population centers.  Each iteration of this approach begins with a set of social media users who have been classified as either being in the location of interest, or not.  In the {\it expand} step, a subset of the users that have been classified as being in the location of interest is selected for expansion queries.  Some or all of the social network neighbors of these users are collected and added to the data set.  This collection not only discovers new users to add to the dataset, but also identifies previously unknown connections between users that are already in the user data set.  

In the {\it classify} step, the location classification of each user is updated using a probability model.  Similar to \cite{compton2014geotagging}, we perform a global optimization to classify each user based on both the user's account information and the classification of the user's neighbors.  To do this, we construct a simple but powerful factor graph model of the social network and find the maximum likelihood location classification according to this model. 

Our approach relies to two very general assumptions.
\begin{enumerate}
	\item We assume that social media users in the same location tend to connect with each other at a higher rate than users in different locations.  We refer to this phenomenon as location \emph{homophily}.  This assumption is supported by the findings of \cite{backstrom2010find}.
	
	\item We assume that certain user profile characteristics, such as the use of location-specific words or phrases, can serve as useful features in classifying a user's location.  This assumption is supported by the location classification models of \cite{han2013stacking} and others.
\end{enumerate}
The first assumption is embedded in our collection methodology.  By collecting the neighbors of users that are identified as being in the target location during the \emph{expand} step, new users in the target location will continue to be identified.  Our classification model makes use of both of these assumptions.

\subsection{Our Contributions}
	In this paper we show how one can efficiently build a set of location specific social media users using an expand-classify approach.  A key element of this approach is our factor graph model for location classification.  We show that classifying users with this model via maximum likelihood reduces to solving a minimum cut problem on an appropriately defined graph.  This makes the classification step very efficient.  Overall, our approach is very simple and requires a minimal amount of user input.
    
    We also empirically demonstrate the efficacy of our approach using data from the social media site Twitter.  
    One of the difficulties inherent in determining social media user locations is the problem of obtaining a labeled dataset.  In order to evaluate the performance of our approach, we use geo-tagged Twitter posts, known as tweets, as a ground-truth.  We find that our method has high accuracy in identifying users from a given location, but that this accuracy decreases for larger population centers.  In addition, we show that by studying the content of location specific users obtained with our approach, we can gain situational awareness in an environment.  In particular, we show we can use our method to identify the onset of significant social unrest in the Philippines.  


\subsection{Previous Work}
Related to the problem of building a set of social media users in a given location is identifying the location
of a given user.  While the former problem has not been addressed in the literature, the latter problem  has been a topic of interest in recent social media research.  Geo-location
of social media users has many real-world applications, including those in emergency response, marketing, law enforcement, military intelligence, and anti-terrorism.  

    \cite{bo2012geolocation} and \cite{han2013stacking} approach the user location problem using user content to identify location.   These methods have been shown to correctly classify 50\%--60\% of user locations on test data.  A downside to using these methods is that they require collection and parsing of each user's posted content, which can become computationally expensive.     
    While many of the user location classification methods in the literature cite online advertising and customized user experience as their motivation for learning user location, some efforts in this area of research have their roots in emergency response, crises, and situational awareness.  \cite{starbird2012learning} uses collaborative filtering and support vector machines to identify Twitter users that are physically present at mass disruption events, such as the Occupy Wall Street protests in New York City in 2011.  \cite{kumar2013whom}, working in a similar vein, introduce a method of identifying users that are providing useful information for gaining situational awareness on the Arab Spring movements in the same year.  This approach combines topic models with user location information to determine user relevance.  

    Other approaches to the user location problem use social network connections.     \cite{Davis-InferringLocation} give a method for inferring Twitter user location that uses declared profile locations, tweet geo-locations, and the locations of each user's friends.  
    \cite{compton2014geotagging} take a similar approach to \cite{davis2012unsupervised}, and use geo-tagged content to evaluate their method.    Many others have adopted similar approaches to identifying user locations from a social media data set.  \cite{jurgens2013s} uses a label propagation algorithm to assign locations to users in a social network based on a few known user locations. \cite{kong2014spot},  propose methods that assign weights to user relationships that quantify their utility in discerning location.  \cite{mcgee2013location} also introduce a model that uses weighted social media relationships to determine user locations, based on a model for predicting online relationship strength proposed by \cite{gilbert2009predicting}. \cite{li2012multiple} give a method for assigning locations based on user behavior likelihood models, and \cite{rout2013s} show that support vector machines can also be used to classify users' locations.
    
     \cite{backstrom2010find} provides a detailed analysis of how distance correlates to online relationships in social media sites.  One of the important findings in this work is that relationships tend to be less geographically localized in more dense population centers.  We make the same observation, and our ability to build sets of users from large metropolises suffers as a result of this characteristic.

    There have also been many efforts in community detection within social media.  Community detection can be thought of as a generalization of the user location problem, and approaches to the two problems often rely on similar assumptions.  \cite{NIPS2012-4532} present an approach that considers both group interconnectedness and similarity in features.  Another well-known approach to community-detection are mixed-membership stochastic blockmodels, introduced by \cite{airoldi2008mixed}, which rely solely on network structure to assign users into community groups.  Both of these community detection methods are unsupervised.  This is in contrast to most location classification methods, which tend to rely on having a set of labeled data.


 \section{Current Approaches }

To build a set of location-based social media users today, one must rely on search functions of the social networks.  In this work we focus on the social network Twitter, so here we will
investigate some of the challenges and limitations of Twitter's search functions.  The main limitation we find is that 
the search functions either return a very limited set of users or a set of users not in the location of interest.    We now present details of our investigation.

    \subsubsection{Twitter User Search}

    The Twitter user search API enables a person to search for Twitter users based on a query string \citep{rest-usersearch}.  The API returns user profiles that contain a match or partial match of the query string in the profile information.  We used this method to produce the seed sets of users for all of our collections. This method can return up to 1000 profiles for a specific search query.

    In several cases, this API did not return any results for specific location queries.  We tried different queries such as ``Corinto, Colombia'', ``Binghamton, NY'', and ``Caracas, Venezuela''. In these cases we used the city or town name only in the query to produce results, with more unique town names producing better seed sets.  The seed set from the ``Casimiro de Abreu'' query included only two accounts that appeared to be in or related to the target location.     Even when the specific location query returned results, the profiles obtained were not necessarily in the target location.  Among the results from the ``Binghamton, NY'' query, for example, was a user from Virginia whose only apparent connection to Binghamton was a claim in his profile description that he had once met his favorite celebrity there.      Using the user search for Caracas returned close to the maximum of 1000 users.  We found geo-located tweets for 102 of these users.  Of these, only 27 were inside of a 15 mile radius around Caracas (Figure \ref{fig: Caracas}), while the remaining 75 were scattered around the world.

    \subsubsection{Twitter Search}

    The Twitter search API, different from the Twitter user search API, returns tweets that contain a match or partial match for a query string.  This API also takes an optional location and radius, and returns only matching tweets originating from inside the radius.  According to the documentation, the API uses tweet geo-tagging if available, but otherwise will use profile location information \cite{rest-search}.    One of the primary drawbacks of attempting to use this method to build a set of users is the need to supply a query string, as this API searches for tweets and not users.  Also, the Twitter search API limits its results to tweets from the previous week, so less active users would not be found using this method.

    We attempted to employ this method by executing a Twitter search query for the single character \emph{e}, and supplying the latitude--longitude location and radius that was used for labeling geo-tagged tweets in each collection.      The users that posted the tweets returned by these queries did not generally appear to be in the corresponding target locations.  For example, executing this query using the location of the greater Binghamton area returned tweets from 66 users, only six of whom indicated they were in the greater Binghamton area.     Many of the remaining 60 user profiles returned by the search query indicated user locations that were clearly not in the target area.      Searching through the 66 user timelines for geo-located tweets only yielded two locations: one inside the Binghamton area and one outside.  


\section{Classification Model}
The limitations of existing approaches for building a set social media users from a given location motivate our expand-classify approach which we will show can quickly
produce a large number of location-specific users.  Key to this approach is the location classification of all users collected in the data set.  
In this section, we present a classification method that uses a factor-graph model based on the image segmentation model presented by \cite{zabih2004spatially} and is closely related to the Ising energy model.  We make a set of minimal assumptions on this model which allow for the classification to be efficiently done via minimum graph-cuts.  We now provide details of the model.

    \subsection{Factor Graph Representation}

        We use a factor graph to serve as a generative model of user locations and connections within a social network.  Nodes in a factor graph represent variables, which can be latent or observed.  Nodes are connected to factors, which imply a dependency structure that specifies a factorization of the joint distribution function of  variables associated with the nodes.

        \subsubsection{Nodes.}  In our graph we consider three types of nodes, representing the three types of variable in our model: 
        \begin{enumerate}
            \item User profile information $\vx_{i}$ for each user $i$. This vector includes information on whether a user's profile information contains location-specific terms.  These values are observed.
            \item Relationship features $\vz_{i,j}$ for pairs of users $i,j$.  This vector encodes the nature of the social media relationship (who is following whom), the out degree of the ``follower,'' and the in-degree of the ``friend.''  These values are also observed.  
            \item User location class $\ell_{i}$ for each user $i$.  These values are unobserved, or latent.
        \end{enumerate}

        \subsubsection{Factors.}
        The factor graph also contains two types of factors:
        \begin{enumerate}
            \item For each user $i$, corresponding nodes $\vx_{i}$ and $\ell_{i}$ share a common factor with potential 
            \[
            f(\vx_{i},\ell_{i})=e^{-\phi(\vx_{i},\ell_{i})}.
            \]  
            \item For each pair of users $i,j$, the corresponding nodes $\ell_{i}$, $\ell_{j}$, and $\vz_{i,j}$ share a factor with potential 
            \[
            g(\vz_{i,j},\ell_{i},\ell_{j})=e^{-\psi(\vz_{i,j},\ell_{i},\ell_{j})}.
            \]  

        \end{enumerate}We refer to $\phi$ and $\psi$ respectively as the \emph{profile energy} and  \emph{link energy} functions.
For pairs of users $i,j$ for which there is no observed social media relationship (encoded in vector $\vz_{i,j}$), we fix $g(\vz_{i,j},\ell_{i},\ell_{j})=1$.  This modeling choice and the assumptions implied by it are discussed in Section \ref{sec: assumptions} below.
        \begin{figure}[!hbt]
	\centering
	\includegraphics[scale=0.5]{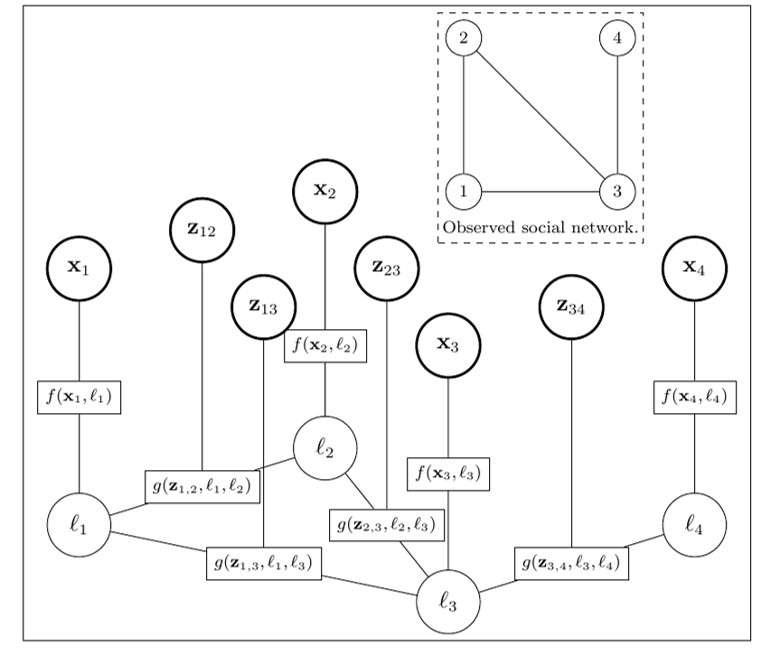}
	            \caption{Factor graph model for social media user location classes.}\label{fig: factor-graph}
\end{figure}

        Figure \ref{fig: factor-graph} provides a notional social network consisting of four users and the corresponding factor graph model.  This figure omits trivial factors with fixed potential functions, i.e., factors corresponding to pairs of users that are not connected to each other in the social network.  The nodes with heavier outlines represent observed values, while the $\ell_{i}$ (location) nodes represent latent variables.  

        Given $N$ observed users in a social network, let $\mathbf{X}$ be all of the users' observed profile features, $\mathbf{Z}$ be all of the observed relationship features, and $\mathbf{L}$ be a vector of latent user location classes.  Our factor graph model implies that the joint probability of these vectors is proportional to the product of the potentials, i.e.,
        \begin{equation}
        \Pr(\mathbf{X},\mathbf{Z},\mathbf{L}) = \frac{\prod_{i=1}^{N}e^{-\phi(\vx_{i},\ell_{i})}\prod_{\{i,j: \; i < j\}}e^{-\psi(\vz_{i,j},\ell_{i},\ell_{j})}}{\mathbf{Z}(\boldsymbol{\Phi,\Psi})}, \label{eq: joint_dist}
        \end{equation}
        where $\mathbf{Z}(\boldsymbol{\Phi,\Psi})$ is the partition function.

    \subsection{Model Characteristics} \label{sec: assumptions}

        We have assumed that user profile characteristics and relationships can be modeled by a probability distribution that factorizes according to the structure of our factor graph representation.  We now provide some additional specifications that support our objective of obtaining a set of users in a specified location and follow from our assumptions.

        \subsubsection{Location Classes.} We assume a two-class location model, in which we only wish to label each user in the dataset as either being in the location of interest or not in the location of interest.  We set the location class variable   $\ell_{i}$ to one if user $i$ is in the location of interest and zero otherwise.
        
        \subsubsection{Link Energy Function.} \label{sec: link energy assmptions} Without yet specifying a link energy function, we impose the following limitations on its structure:
        \begin{enumerate}
            \item We assume that the probability of an edge forming between two users in a social network is small, irrespective of whether or not they are in the same location.  To demonstrate the implications of this assumption, let $A_{ij}$ be the event that a relationship exists between user $i$ and user $j$ on a social network.  Our assumption implies that
            $\Pr(A^c_{ij}|\ell_{i} \neq \ell_{j})\approx            \Pr(A^c_{ij}|\ell_{i}=\ell_{j}) \approx 1$.  In other words, the location labels of do not have a strong effect on the probability of the non-existence of a social network connection.
            Based on this assumption, we set $\psi(\vz_{ij},\ell_{i},\ell_{j})=0$ and $g(\vz_{ij},\ell_{i},\ell_{j})=1$ when there is no observed connection between users $i$ and  $j$.  
            
            \label{item: no relationship assumption}
            \item \label{item: inequalities assumption} We assume location homophily, which means that social network links between two users that are in different location classes are always less probable than links between users in the same location class.  This implies that for any users $i$ and $j$, the following inequalities hold:
            \begin{align*}
                \psi(\vz_{ij},1,1) \leq \psi(\vz_{ij},0,0) &\leq \psi(\vz_{ij},0,1)\\
                \psi(\vz_{ij},1,1) \leq \psi(\vz_{ij},0,0) &\leq \psi(\vz_{ij},1,0).
            \end{align*}
			  We allow for a positive cost to be associated with classifying two connected users as both being outside of the location of interest, as this classification does not necessarily imply that they are in the same location.  However, we restrict this cost to be no more than the cost of assigning different location classes to a pair of connected users.   By convention, we set $\psi(\vz_{ij},1,1) = 0$.  An important implication that follows from this assumption is
            \[
                \psi(\vz,1,1) + \psi(\vz,0,0) \leq \psi(\vz,1,0) + \psi(\vz,0,1).
            \]
            These inequalities allow for efficient location classification using graph cuts, as shown by \cite{kolmogorov2004energy}.
        \end{enumerate}

    \subsection{Classification Optimization} \label{sec: optimization}

        We set as our objective to find the most probable location classifications, $\mathbf{L}$.  More formally, having observed values $\mathbf{X}$ and $\mathbf{Z}$, we seek a solution to the following optimization problem:
        \[
        \maximize_{\mathbf{L}}~\Pr(\mathbf{X},\mathbf{Z},\mathbf{L}),
        \]
        which is equivalent to finding a solution to the following:
        \begin{equation}
            \minimize_{\mathbf{L}}  \sum_{i}\phi(\vx_{i},\ell_{i}) + \sum_{i < j}\psi(\vz_{i,j},\ell_{i},\ell_{j}). \label{eq: energy_optimization}
        \end{equation}
        \cite{boykov2001fast} provide an efficient method for minimizing functions of this nature using graph cuts.  Following their method, and consistent with the subsequent findings by \cite{kolmogorov2004energy}, we construct a directed graph consisting of a source node $s$, a sink node $t$, and a node $u_{i}$ for each user $i$.  We add the following edges and capacities:
        \begin{itemize}
            \item An edge from each user node $u_{i}$ to the sink node $t$ with capacity 
            \[
            c_{(u_{i},t)}=\phi(\vx_{i},1).
            \]
            \item For each pair of users $i,j$ for which there is an observed relationship in the social network, edges from $u_{i}$ to $u_{j}$ and from $u_{j}$ to $u_{i}$ with capacities
            \begin{align*}
            c_{(u_{i},u_{j})} &=\psi(\vz_{i,j},1,0) - \frac{1}{2}\psi(\vz_{i,j},0,0), \\
            c_{(u_{j},u_{i})} &=\psi(\vz_{i,j},0,1) - \frac{1}{2}\psi(\vz_{i,j},0,0).
            \end{align*}
            \item An edge from the source node $s$ to each user node $u_{i}$ with capacity 
            \[
            c_{(s,u_{i})}=\phi(\vx_{i},0)+\frac{1}{2}\sum_{\{j:j\neq i\}}\psi(\vz_{i,j},0,0).
            \]
        \end{itemize} 
        We refer to this graph as the {\bf Energy Graph} representation of the energy function
        \begin{equation}
            E(\mathbf{L})
            = \sum_{i}\phi(\vx_{i},\ell_{i}) + \sum_{i < j}\psi(\vz_{i,j},\ell_{i},\ell_{j}), \label{eq: energy_function}
        \end{equation}
        which is the objective function in the classification optimization,  equation \eqref{eq: energy_optimization}.  An example energy graph for the social network and factor graph in Figure \ref{fig: factor-graph} is shown in Figure \ref{fig: min-cut-graph}.  The following result, which is proved in Appendix \ref{sec:proof_thm},  shows that performing the optimization in equation \eqref{eq: energy_optimization} is equivalent to finding the minimum capacity cut on the Energy Graph, which can be done efficiently using minimum cut-maximum flow algorithms \citep{boykov2001fast}.
        \begin{figure}[!hbt]
            \centering
            \begin{tikzpicture}
                \node[circle,draw]
                    (u1) at (0,0){$u_{1}$};
                \node[circle,draw]
                    (u2) at (4.8,1.5){$u_{2}$};
                \node[circle,draw]
                    (u3) at (7.2,-1.5){$u_{3}$};
                \node[circle,draw]
                    (u4) at (12,0){$u_{4}$};
                \path[draw,-latex] 
                    (u1.north east)
                     -- 
                     node[above,pos=0.6,rotate=17]{\tiny $\psi(\vz_{1,2},1,0)-\frac{1}{2}\psi(\vz_{1,2},0,0)$}
                     (u2.west);
                \path[draw,-latex] 
                    (u2.south west)
                     -- 
                     node[below,pos=0.4,rotate=17]{\tiny $\psi(\vz_{1,2},1,0)-\frac{1}{2}\psi(\vz_{1,2},0,0)$}
                    (u1.east);
                \path[draw,-latex] 
                    (u3.south west) 
                     -- 
                     node[below,pos=0.5,rotate=-11.5]{\tiny $\psi(\vz_{1,3},0,1)-\frac{1}{2}\psi(\vz_{1,3},0,0)$}
                    (u1.south);
                    \path[draw,-latex] 
                    (u1.south east) 
                     -- 
                     node[above,pos=0.45,rotate=-11.5]{\tiny $\psi(\vz_{1,3},1,0)-\frac{1}{2}\psi(\vz_{1,3},0,0)$}
                    (u3.west);
                \path[draw,-latex] 
                    (u3.north west) 
                     -- 
                     node[above,pos=0.5,rotate=-51]{\tiny \makebox[12pt] {~\hspace{6pt}\raisebox{10pt}{$\psi(\vz_{2,3},0,1)$}}$-\frac{1}{2}\psi(\vz_{2,3},0,0)$}
                    (u2.south east);
                \path[draw,-latex] 
                    (u2.south) 
                     -- 
                     node[below=-6pt,pos=0.7,rotate=-52]{\tiny \makebox[0pt]{$\psi(\vz_{2,3},1,0)$}\raisebox{-9pt}{~\hspace{-1.5em}$-\frac{1}{2}\psi(\vz_{2,3},0,0)$}}
                    (u3.west);
                \path[draw,-latex] 
                    (u3.north east) 
                     -- 
                     node[above,pos=0.5,rotate=18]{\tiny $\psi(\vz_{3,4},1,0)-\frac{1}{2}\psi(\vz_{3,4},0,0)$}
                    (u4.north west);
                \path[draw,-latex] 
                    (u4.west) 
                     -- 
                     node[below,pos=0.5,rotate=18]{\tiny $\psi(\vz_{3,4},0,1)-\frac{1}{2}\psi(\vz_{3,4},0,0)$}
                    (u3.east);
                \node[circle,draw,very thick]
                    (s) at (6,4){$s$};
                \path[draw,-latex] 
                    (s.south west) 
                     -- 
                     node[above,pos=0.4,rotate=33]{\tiny $\phi_{\vx_{1},0}+\frac{1}{2}\sum_{j\ne 1}\psi(\vz_{1,j},0,0)$}
                    (u1.north);
                \path[draw,-latex] 
                    (s.south west) 
                     -- 
                     node[above,pos=0.5,rotate=64]{\tiny $\phi_{\vx_{2},0}$}
                    (u2.north);
                \path[draw,-latex] 
                    (s.south west) 
                     -- 
                     node[below,pos=0.5,rotate=64]{\tiny $+\frac{\sum_{j\ne 2}\psi(\vz_{2,j},0,0)}{2}$}
                    (u2.north);
                \path[draw,-latex] 
                    (s.south east) 
                     -- 
                     node[above,pos=0.5,rotate=-78]{\tiny $\phi_{\vx_{3},0}+\frac{1}{2}\sum_{j\ne 3}\psi(\vz_{3,j},0,0)$}
                    (u3.north);
                \path[draw,-latex] 
                    (s.south east) 
                     -- 
                     node[above,pos=0.4,rotate=-33]{\tiny $\phi_{\vx_{4},0}+\frac{1}{2}\sum_{j\ne 4}\psi(\vz_{4,j},0,0)$}
                    (u4.north);
                \node[circle,draw,very thick]
                    (t) at (6,-4){$t$};
                \path[draw,-latex]  
                    (u1.south east) 
                     -- 
                     node[below,pos=0.7,rotate=-33]{\tiny $\phi(\vx_{1},1)$}
                    (t.north west);
                \path[draw,-latex]  
                    (u2.south) 
                     -- 
                     node[below,pos=0.8,rotate=-78]{\tiny $\phi(\vx_{2},1)$}
                    (t.north);
                \path[draw,-latex]  
                    (u3.south) 
                     -- 
                     node[below,pos=0.5,rotate=63]{\tiny $\phi(\vx_{3},1)$}
                    (t.north);
                \path[draw,-latex]  
                    (u4.south west) 
                     -- 
                     node[below,pos=0.7,rotate=33]{\tiny $\phi(\vx_{4},1)$}
                    (t.north east);
            \end{tikzpicture}
            \caption{Energy Graph representation of the energy equation corresponding to the factor graph in Figure \ref{fig: factor-graph}.}\label{fig: min-cut-graph}
        \end{figure}
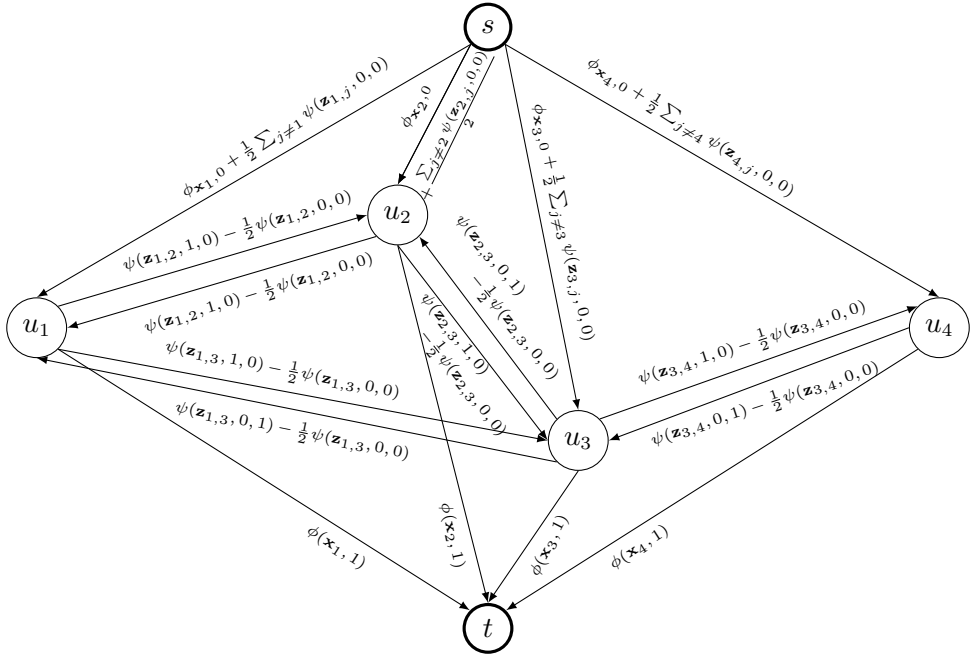

        \begin{theorem}[Minimum Cut-Classification Optimality Equivalence] \label{thm: optimality}
            Given an energy minimization of the form given in \eqref{eq: energy_optimization}, the optimal solution corresponds exactly to a minimum capacity $s$-$t$ cut in the corresponding Energy Graph representation.
        \end{theorem}

Theorem \ref{thm: optimality} greatly simplifies the classification step and is a result of the constraints we impose on the link functions, which are quite general.

%
%

\section{Choosing the Energy Functions}

    We now discuss our choices for the energy functions $\psi$ and $\phi$.  These functions quantify the trade-off between the value of the information contained in a user's profile (e.g., the user's self-identified ``location'') and value of the user's social connections.  For our implementations, we use a fixed link energy function that follows from the findings of \cite{backstrom2010find}, \cite{mcgee2013location}, and \cite{davis2012unsupervised}.  For the profile energy function, we compare two approaches: a naive approach in which we fix the profile energy based on qualitative observations, and a parametric approach in which we fit a probabilistic model to a subset of the data for which location labels are available.  

    \subsection{Link Energy Function}

    Social media research has consistently shown that users tend to connect to other users with whom they have an existing relationship outside of social media (see, e.g., \cite{backstrom2010find}, \cite{Davis-InferringLocation}) and that relationships between users with lower degrees tend to be indicative of closer relationships \citep{gilbert2009predicting}.  Based on these findings we assume that the utility of a social media relationship in inferring that two users belong to the same location decreases as the number of relationships (or degree) of linked users grows.  If we observe that a user is following a superstar with millions of social media connections, for example, we would not consider that online relationship to be very valuable in location inference.  On the other hand, if two users are connected and each has only a total of 20 online connections, we consider that relationship to be indicative of an existing relationship outside of social media, which could mean that the users live in close proximity to each other.

    Suppose user $1$ is following user $2$ in a directed social network.  We also observe that user $1$ follows a total of $z_{1}$ other users and that user $2$ has a total of $z_{2}$ followers.  We encode this information in vector $\vz_{1,2}$, and use a sigmoid function to model our intuition on link energy, setting
    \[
    \psi(\vz_{1,2},1,0)=\psi(\vz_{1,2},0,1) = \frac{\gamma}{1+\exp(-2+(2/\alpha_{1})z_{1}+(2/\alpha_{2})z_{2})}.
    \]
    Using this form, the parameters $\gamma$, $\alpha_{1}$, and $\alpha_{2}$ all have useful interpretations.  Parameter $\gamma$ is the link energy of the closest relationships.  If user 1 has very few friends and user 2 as very few followers, this function approaches $\frac{\gamma}{1+\exp(-2)}\approx\gamma$.  The parameters $\alpha_{1}$ and $\alpha_{2}$ are the numbers of friends and followers, respectively, that would result in half this link energy.  

    Based on the findings of \cite{mcgee2013location} and our own investigation of Twitter relationships, we fixed $\alpha_{1}=500$ and $\alpha_{2}=5000$.  We have observed that users with more than about 500 friends tend to be connected to more celebrities, politicians, and media sites, while users with more than about 5000 followers tend to start having more than just a local following.  Figure \ref{fig: psi_decay} illustrates how this function decays as the degree of each node in a social media relationship increases.  This figure illustrates the energy of a directed relationship in which user 1 is following user 2. In the left-hand plot, user 1's out-degree, or friends count, is fixed at 20.  We see that if user 2 has close to 0 followers, the link energy is close to $\gamma$, but decays as the number followers grows.  The right-hand plot shows the same effect as the number of user 1's friends increases, while user 2's follower count is fixed at 20.

    \begin{figure}
    \centering
    \begin{tabular}{cc}
    \includegraphics[width=0.45\textwidth]{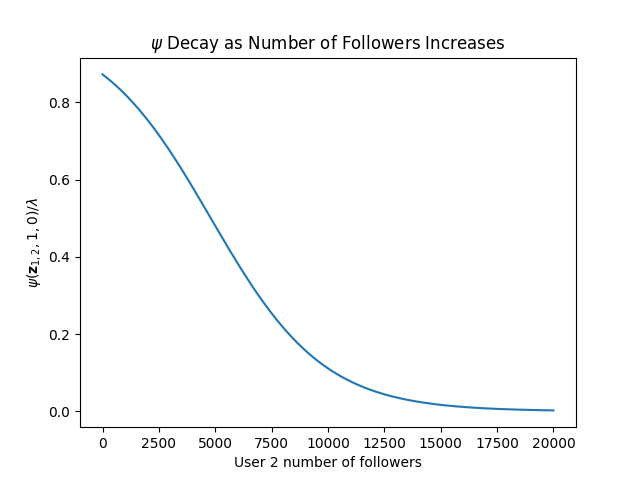} & \includegraphics[width=0.45\textwidth]{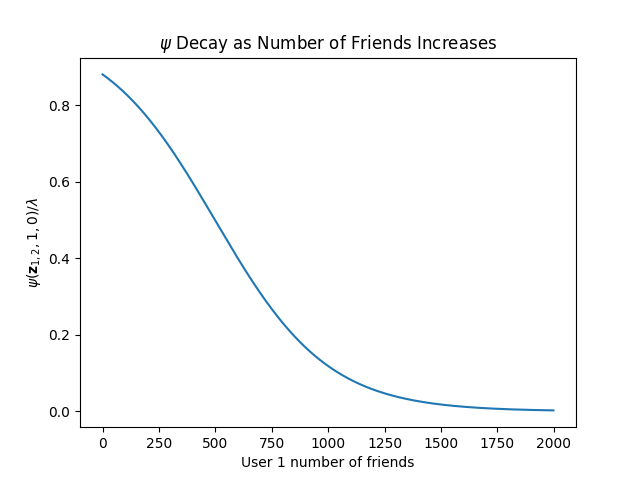} \\
    $z_{1}=20$ & $z_{2}=20$
    \end{tabular}
    \caption{Decay of link energy $\psi(\vz_{1,2},1,0)$ as the number of user 1 friends or number of user 2 followers increases.} \label{fig: psi_decay}
    \end{figure}

    The function $\psi(\vz_{ij},1,0)$ can be interpreted as a log likelihood ratio.  Noting that we have set $\psi(\vz_{ij},1,1)=0$, we can express the link energy as 
    \begin{align*}
    \psi(\vz_{ij},1,0)&=\log\left(\frac{e^{-\psi(\vz_{i,j},1,1)}}{e^{-\psi(\vz_{ij},1,0)}}\right).
    \end{align*}
    This is the log ratio of factor potentials from our factor graph model.  This can be thought of as the log likelihood ratio that an observed relationship is indicative of two users sharing the same location.  This interpretation is useful in considering our choice for the parameter $\gamma$.  In our implementations, we initially set $\gamma=\log (5)$, which implies that low-degree relationships are about five times more likely to share a common location than not.  We show through sensitivity analysis that this achieves good performance in many cases.

    We have addressed the link energy value for connected users when both users are in the location of interest ($\psi(\vz_{ij},1,1)$), and when one user is in the location of interest and the other is not ($
    \psi(\vz_{ij},1,0),\psi(\vz_{ij},0,1)$).
    We still have to address the link energy value when the pair of connected users is not in the location of interest ($\psi(\vz_{ij},0,0)$).  Because our approach continues to collect friends and followers from users classified within the location, we do not expect to obtain many, or perhaps even any edges between user pairs in which both users are outside of the location of interest.  By assumption, 
    \[
    \psi(\vz_{i,j},1,1)\leq\psi(\vz_{i,j},0,0)\leq\psi(\vz_{i,j},1,0)=\psi(\vz_{i,j},0,1).
    \]

    Given these bounds, where we set link energy $\psi(\vz_{i,j},0,0)$ in this range provides for interesting discussion.  On one hand, we can set $\psi(\vz_{i,j},0,0)=\psi(\vz_{i,j},1,1)=0$, arguing that users assigned to the same location class should always have zero link energy.  However, this fails to recognize that unlike user pairs in location class 1, two users in location class 0 do not necessarily share the same geographic location.  It follows that relationships between two users in location class 0 should be associated with some positive energy, implying that they are less probable than relationships between users in location class 1.
    Therefore, we assume that the link energy between users in location class 0 is very close to the link energy between users in different location classes.  Specifically, we set
    \[
    \psi(\vz_{i,j},0,0)=\lambda\psi(\vz_{i,j},1,0),
    \]
    where $\lambda$ is close to, but less than 1.  We provide sensitivity analysis of this decision using other values of $\lambda \in [0,1]$.

    We can think of $\lambda$ as a way of dampening the effect of relationships with users in location class 0.  If $\lambda=0$, then a user $i$ with many strong relationships with other users in location class 0 will be ``pulled'' into location class 0 with them.  If $\lambda=1$, user $i$'s connections with users in location class 0 will not have a direct affect on $i$'s classification, because these relationships will result in the same link energy in either case.

    \subsection{Profile Energy Function}

    We adopt two different approaches for coming up with a profile energy function.  In the first case, we assume that there is no labeled data available.  In this case, an analyst can rely on observations, expert information, or intuition to construct a simple and yet potentially powerful profile energy model.  Alternatively, if there is some labeled data available or if a feasible method exists for labeling some of the data as it is collected, an analyst can fit a parametric probability model to the labeled data and use this model to construct a link energy function.

    Just as the link energy function has a probabilistic interpretation as a log likelihood ratio, the link energy function has an analogous interpretation.  Specifically,
    \[
        \phi(\vx_{i},0)-\phi(\vx_{i},1)
        =\log \frac{
            e^{
                -\phi(\vx_{i},1)
            }
        }{
            e^{
                -\phi(\vx_{i},0)
            }
        }.
    \]
    Because the factor potentials can be scaled without affecting the joint probability distribution, the scaling will be absorbed into the partition function in equation \eqref{eq: joint_dist}. Only the difference between location potentials, and not the values themselves, are relevant in optimizing user $i$'s location assignment.  Rescaling the energy functions is equivalent to adding a constant to the objective function in equation \eqref{eq: energy_optimization}.  Therefore, when determining a profile energy function, one need only be concerned with the difference, $\phi(\vx_{i},0)-\phi(\vx_{i},1)$ for each user $i$, recognizing that this difference represents the log likelihood ratio of the location classes.

    \subsubsection{Fixed Profile Energy Model} \label{sec: fixed-model}

    In the absence of labeled data, an analyst could use intuition, expert information, or observations to produce a simple odds table, from which the profile energy function could be produced.  For example, suppose the analyst is interested in finding users in Boston, Massachusetts.  The analyst might decide observing the word ``Boston'' in a user's profile location field is a useful feature in this classification.  Therefore, the analyst could simply conjecture odds for each feature category, such as those given in Table \ref{table: example_odds}.

    \begin{table}[!hbt]
    \centering
    \fontsize{11}{18}\selectfont
    \caption{Example odds tabled used to construct a profile energy function.} \label{table: example_odds}
    \vspace{6pt}
    \begin{tabular}{l|cc}
    Feature & Odds (location:non-location) & $\phi(\vx_{i},0)-\phi(\vx_{i},1)$ \\ \hline
    Profile loc. includes ``Boston'' & 20:1 & $\log(20)$ \\
    Profile loc. does not include ``Boston'' & 1:10 & $-\log(10)$ \\[12pt]
    \end{tabular}
    \end{table}

    In practice we found that this relatively naive approach to constructing a link energy function achieves performance similar to the parametric approach described below.  

    \subsubsection{Parametric Approach} \label{sec: logistic_regression}

    The drawback of the naive approach is that it requires an analyst to come up with a set of location features and a location odds table based on those features.  This task can become very difficult as the number of features increases.  However, there might be very many features that are useful in user location classification, and these features might not be mutually exclusive.  In this case, a parametric model would be useful in constructing the profile energy function.  Because of its simplicity, we propose a linear model:
    \[
    \phi(\vx_{i},0)-\phi(\vx_{i},1) = \beta^{T}\vx_{i}.
    \]
    Using our interpretation of the profile energy difference, $\phi(\vx_{i},0)-\phi(\vx_{i},1)$, as a log likelihood ratio, this linear model is the well-known logistic regression model and is easily fit on a set of labeled data using existing methods and open source software packages.  We fit a regularized logistic regression model, which finds parameters $\beta$ by performing the following optimization:
    \[
    \maximize_{\beta}~C\left(\sum_{i: \; \ell_{i}=1}\log\left(\frac{1}{1+\exp(-\beta^{T}\vx_{i})}\right)+\sum_{j:\; \ell_{j}=0}\left(\frac{1}{1+\exp(\beta^{T}\vx_{j})}\right)\right) - \|\beta\|,
    \]
    where $C$ is the regularization parameter and $\|\beta\|$ is the regularization norm ($\sum|\beta|$ for $L1$ regularization and $\frac{1}{2} \beta^{T}\beta$ for $L2$ regularization) \citep{scikit-learn}.

\section{Implementations}

    In this section we provide an analysis of the results obtained by applying our expand-classify approach to collect a set of users from different locations.  We choose several locations where the primary language is not English and where we do not have any prior knowledge about local trends and customs to demonstrate the power of our approach.   The locations include small and large cities from all over the world.

    For each location we provide a brief overview, a summary of the composition of the set of seed users, implementation details, and an analysis of the results obtained. We use the geo-tagged posts of users as a ground truth location label in order to evaluate the accuracy of our approach.  For one of the locations we analyze the content of the users to show that we can detect the onset of political unrest.
    
    In order to obtain a set of seed users for each location, we used the Twitter ``user search'' API which allows queries for users meeting certain criteria \citep{rest-usersearch}.  In some cases, this method did not return any results for a specific location, and a more general location query string was used.  The nature of the resulting seed set will be provided for each location.

    \subsection{Corinto, Colombia}

    Corinto, Colombia is a town in the Cauca district of Colombia, located about 30 miles southeast of Cali.  Including the population of its nearby and larger neighbor Miranda, the Corinto area has a population of approximately 30,000 people.    Using Google Maps \citep{googlemaps}, we located its center at 3.174159\degree N, 76.25880\degree W.  For labeling geo-tagged tweets, we used a radius of 7 miles (see Figure \ref{fig: Corinto}).  This radius includes Miranda as well as some of the smaller nearby towns but does not include any part of the Cali metropolis.

    \begin{figure}[!htb]
    \centering
    \includegraphics[scale=0.33]{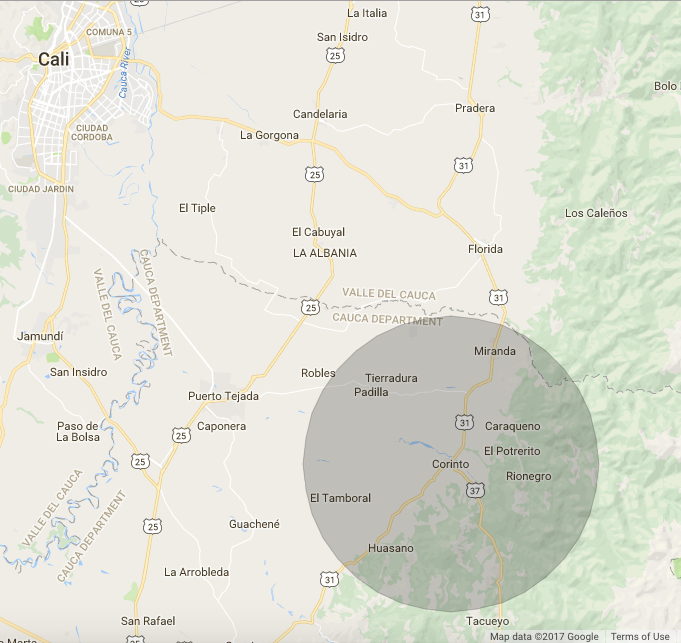}
    \caption{Corinto, Colombia label radius, plotted on Google Maps \citep{googlemaps}.}\label{fig: Corinto}
    \end{figure}

    \subsubsection{Seed Set}

    Querying the Twitter user search API for ``Corinto, Colombia'' to obtain a seed user set did not return any results.  Instead, we obtained 128 users returned from running the individual queries ``Corinto'' and ``Colombia'' in the user search API.  Of these results, 67 profiles contained the word ``Corinto'' in the location, description, name, or screen name, while 62 of the results contained the word ``Colombia'' in at least one of these four fields.  Only one result contained both strings.  

    Searching through the tweets from these accounts yielded 12 geo-located tweets; of these only one was inside the 7 mile radius depicted in Figure \ref{fig: Corinto}.  These locations were consistent with a manual inspection of the seed accounts, which included profile locations from varying locations throughout Colombia and from around the world.  The seed set appeared to contain very few accounts in the target location.

    \subsubsection{Logistic Regression Energy Model}

    In order to fit a logistic regression model on the data, we developed a method of extracting features from user profiles that might be useful in predicting the user's location classification.  First, we created two lists of character strings, $W_{1}$ and $W_{2}$, which we compared to each user's profile information.  List $W_{1}$ was comprised of strings that we thought might indicated a user was associated with the target location, while $W_{2}$ contained strings that would suggest a user was not associated with the target location.  Each string from list $W_{1}$ was used to generate four binary feature variables, corresponding to the user's profile location, description, name, and screen name fields.  The character strings that comprise list $W_{1}$ for Corinto are in Table \ref{table: Corinto_terms}.

    \begin{table}[!hbt]
    \fontsize{11}{18}\selectfont
    \centering
    \caption{List of character strings $W_{1}$ used to extract profile features for Corinto logistic regression.} \label{table: Corinto_terms}
    \vspace{6pt}
    \begin{tabular}{llll}
        ``Corinto'' & 
        ``Cauca'' & 
        ``Colombia'' & 
        ``Miranda'' \\
        ``Corinto Colombia'' & 
        ``Corinto, Colombia'' & 
        ``Miranda, Colombia'' &
        ``Miranda Colombia''\\
        ``Corinto Cauca'' & 
        ``Corinto, Cauca'' & 
        ``Miranda, Cauca'' & 
        ``Miranda Cauca'' 
    \end{tabular}
    \end{table}

    Because there are 12 character strings in this list, there were 48 corresponding binary variables in the logistic regression model.  We include strings containing the location Miranda, Colombia in list $W_{1}$ because Miranda is a population center within the 7-mile radius of Corinto.

    List $W_{2}$ simply contained a list of world cities with populations over 1,000,000 from \cite{maxmind}.  This list is contained in Appendix \ref{appendix: world_cities}.  The list $W_{2}$ generates four additional binary variables: one for each user profile field (location, description, name, and screen name). If any of these fields contained a string from $W_2$, the feature value is set to one.

    We included five additional binary feature variables: empty location, language, UTC offset, protected account, and verified account.  The empty location variable took value one if the user profile's location field was left empty.  If a user's profile language was set to the local language, Spanish, then the language variable was set to one.  If the profile's time zone matched the local UTC offset, -18000 seconds, the UTC offset variable was set to one.  The protected and verified account variables were set to match each user's account settings, taking value one if the profile was protected or verified, respectively, and zero otherwise.    
    
    The total number of features in the logistic regression model for Corinto is 57.    For the response variable, we used geo-located tweets posted by the users.  We searched through users' most recent posts and identified any tweets that contained geo-location data.  Of these, we extracted the post that had coordinates closest to the center of Corinto.  If these coordinates were within seven miles of the grid coordinates at the center of the target location, the user was labeled as being inside the target location ($\ell_{i}=1$),  otherwise the user was labeled as outside of the target location ($\ell_{i}=0$).  Users with no geo-tagged tweets were not included in the logistic regression model.

    In fitting the logistic regression models, we set aside some geo-located content for validation and testing.  We used $L1$ regularization and, through validation, found the model achieved the best performance using a regularization coefficient of approximately $C=1$.  

    \subsubsection{Performance}

    Beginning with the seed user set, we iterated the expand-classify approach for four hours.  In the \emph{expand} step of each iteration, we randomly selected up to 30 users from the set of users classified in the target location and queried up to 5000 of each user's Twitter followers.  We randomly selected another 30 users from the set of those classified in the target location and queried up to 5000 of each user's Twitter friends.  We used these values to expand the set of profiles efficiently while remaining within the API rate limits established by Twitter \citep{rest}.

    Following the collections in each iteration, all of the users in the dataset would be classified.  First, the $L1$-regularized logistic regression model was fit and validated on a randomly selected subset of the geo-located users.  Using the resulting linear model as a profile energy function, all of the users in the data set were then classified by finding the minimum cut on the Energy Graph as described in Section \ref{sec: optimization}.  After completing this classification, the \emph{expand} step in the next iteration would begin.

    After four hours, the number of user profiles collected was 140,571.  Of these, 988 were classified as being in the Corinto, Colombia target area.   From the classification results, we constructed a local probability for each user being in the target location:
    \[
        P_{1}(i)=
            \left(
            1+\exp\left(
                \phi(\vx_{i},1)
                -\phi(\vx_{i},0)
                +\sum_{j\neq i}\left[
                \psi(\vz_{ij},1,\ell_{j})
                -\psi(\vz_{ij},0,\ell_{j})
                \right]
            \right)
            \right)^{-1}
    \]
    The probability follows from our factor graph model, holding all location classifications fixed and examining the probability associated with each location class for user $i$.  Using these local probabilities, we construct a Receiver-Operator Characteristic (ROC) curve to evaluate the results on the set of users for which geo-located tweets were available.  The ROC for this collection of Corinto users is plotted in Figure \ref{fig: Corinto-LR-ROC}.  The figure shows that this implementation can correctly classify about 80\% of the users in the target location radius while maintaining a low false positive rate.  The area under the ROC curve (AUC) is 0.92, which is very near the maximum value of one, indicating that the method has good accuracy.   

    \begin{figure}[!htb]
    \centering
    \includegraphics[scale=0.6]{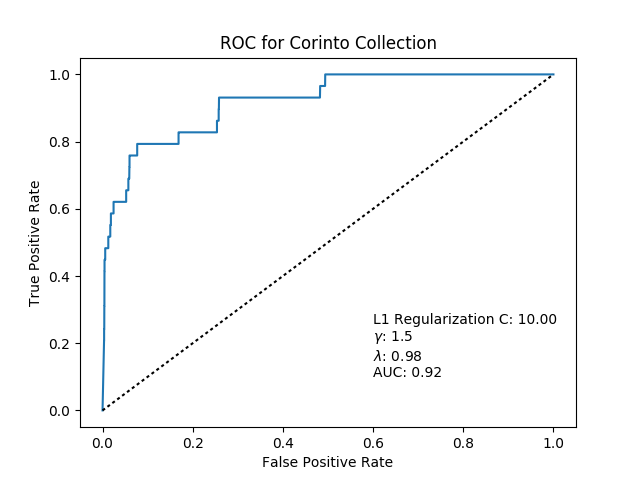}
    \caption{Corinto user classification ROC using logistic regression model for profile energy.} \label{fig: Corinto-LR-ROC}
    \end{figure}

%

    \subsubsection{Sensitivity Analysis}

    We now briefly discuss and illustrate this implementation's sensitivity to the inputs $\gamma$, $\lambda$, as well as the logistic regression regularization.  The parameter $\gamma$ serves as the magnitude of the sigmoid curve that governs the decay of link energy as the number of friends and followers increases.  Higher values of $\gamma$ result in larger link energies, which cause network connections to have more influence over each user's classification.  

    Figure \ref{fig: Corinto-gamma-lambda-sens} depicts the classification ROC curve plotted using several different values for $\gamma$.  Of particular note is the case where $\gamma=0$, which recovers the logistic regression classification without any network information.  Comparison with this curve provides a quantification of the utility of the network structure in this classification model.  Based on the AUC metric, we find that optimal performance appears to occur for higher values of $\gamma$, with $\gamma=\log(10)$ producing an AUC of 0.94.  However, small variations from this value do not appear to substantially impact performance.  Using only the logistic regression model ($\gamma=0$) produces an AUC of approximately 0.64, showing that accounting for network connections in the model substantially improves classification performance.

    The parameter $\lambda \in [0,1]$ sets the link energy for users that are both classified in location set 0 (outside of the target location).  We have set this parameter to $0.98$, so that these relationships are approximately the same cost as relationships for which one user is in the target location and the other is not.  Figure \ref{fig: Corinto-gamma-lambda-sens} depicts the sensitivity of the ROC curve to this value.      Higher values of $\lambda$ appear to produce the best results, and very low values performing very poorly.  This poor performance results from users in the target location being misclassified as location 0 as a result of relationships with other users in location class 0.  Good performance is maintained for values of $\lambda > 0.75$.
  \begin{figure}[!htb]
	\centering
	\includegraphics[scale=0.5]{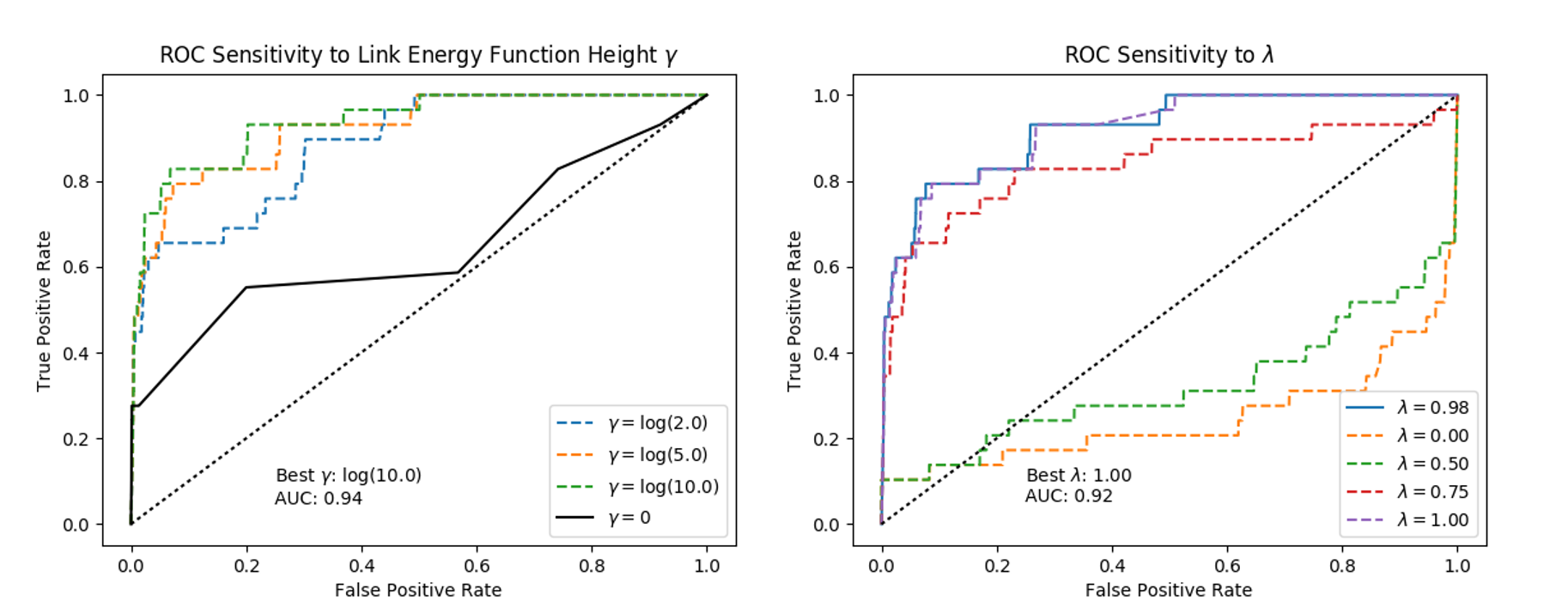}
	\caption{(left) Sensitivity of Corinto user classification to parameter $\gamma$. (right) Sensitivity of Corinto user classification to parameter $\lambda$.} \label{fig: Corinto-gamma-lambda-sens}
\end{figure}


    Finally, we investigate the model sensitivity to the logistic regression regularization.  In a similar fashion to the above analyses, we fit $L1$ and $L2$ regularized logistic regression models to the geo-located training data using different value for the regularization coefficient.  Without fitting the regularization coefficient through model validation, we applied the resulting linear function directly as the profile energy model.  We found that the model performance was neither sensitive to the regularization norm ($L1$ vs. $L2$) nor to the regularization coefficient, except for in cases in which we significantly over-regularized the logistic regression.      Figure \ref{fig: Corinto-LR-sens} shows the classification ROC using for several regularization constants for both $L1$ and $L2$ regularization.  We observe that a lower value of $C$, which implies a more regularized model, results in a slight increase in performance from the value found through validation.

    \begin{figure}[!htb]
    \centering
    \begin{tabular}{cc}
    \includegraphics[width=0.45\textwidth]{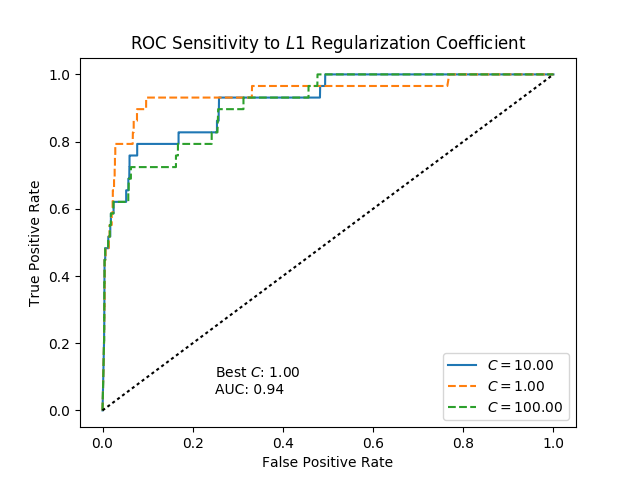}
    &
    \includegraphics[width=0.45\textwidth]{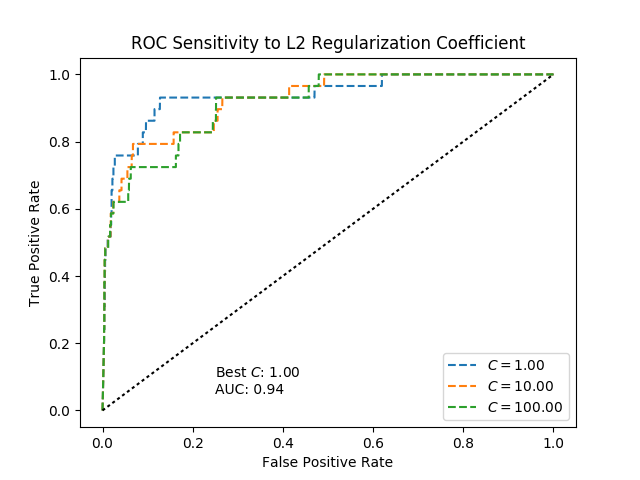}\\
    $L1$ Regularization & $L2$ Regularization
    \end{tabular}
    \caption{Sensitivity of Corinto user classification to logistic regression regularization.} \label{fig: Corinto-LR-sens}
    \end{figure}

    \subsubsection{Summary of Corinto Collection}
    The user collection for Corinto provides a useful example of the utility of the expand-classify approach.  Using the output of a logistic regression model as a profile energy function in our factor-graph model produced a classifier with an AUC of 0.92.  Tuning the parameters on test data enables an increase in performance to an AUC of 0.94.  Our method is superior to the current capability of Twitter's user search API, which in this case was not very useful in producing even a seed set of users in Corinto, Colombia. 


    \subsection{Casimiro de Abreu, Brazil}

    Casimiro de Abreu, Brazil is a town in the Rio de Janeiro state of Brazil, located at 22.484\degree S, 42.202\degree W, about 80 miles east of the city of Rio de Janeiro. It has a population of approximately 35,000%
    .  For labeling geo-tagged tweets, we used a radius of 5 miles (see Figure \ref{fig: Casimiro}).  Based on the imagery available on Google Maps, there are no substantial population centers within this radius.  Casimiro de Abreu falls in the costal region of Barra de S\~ao Jo\~ao.

    \begin{figure}[!htb]
    \centering
    \includegraphics[scale=0.3]{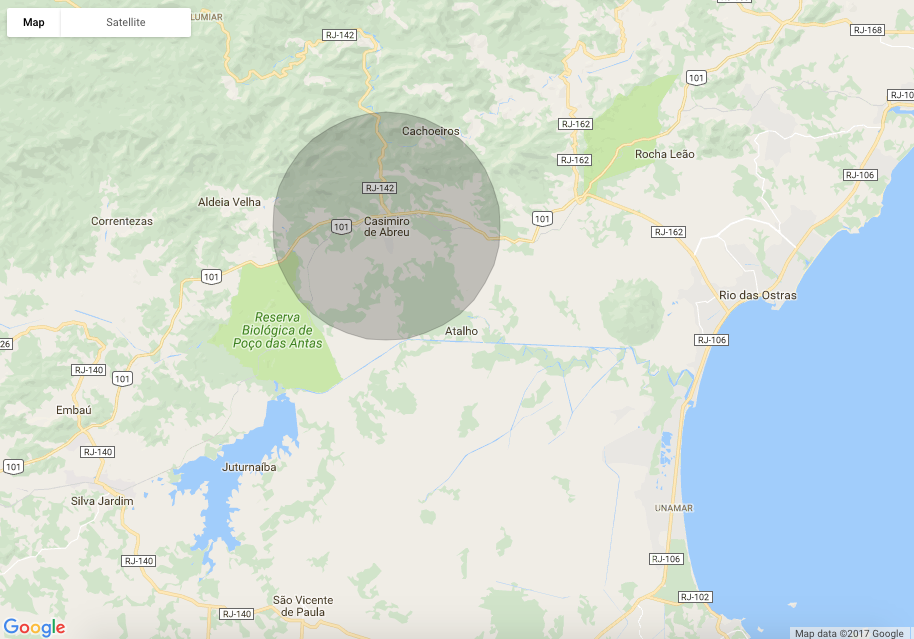}
    \caption{Casimiro de Abreu, Brazil label radius, plotted on Google Maps \citep{googlemaps}.}\label{fig: Casimiro}
    \end{figure}

    \subsubsection{Seed Set}

    Querying the Twitter user search API for ``Casimiro de Abreu, Brazil'' did not return any results.  Instead, we used 11 user profiles returned from running a user search query on ``Casimiro de Abreu.''  Of these results, 10 profiles contained the string ``Casimiro de Abreu'' within the profile information fields, and two of them also contained the string ``Brasil''.  Several of the profiles contained the string ``Casimiro de Abreu'' in a way that did not necessarily refer to a location, and at least two of the profiles indicated a locations that were outside of Brazil.  None of these accounts had recent tweets with geo-location information.  As in the Corinto collection, this seed set did not appear to contain many tweets from the target area.

    \subsubsection{Fixed Profile Energy Function}

    Attempting to use logistic regression as a profile energy model failed to produce a useful set of users from the target area.  The reason for this failure is that there were not enough geo-located users in each iteration inside the target radius to fit a reliable logistic regression model.  Poor classifications in each iteration resulted in more collections of users outside the target region in follow-on iterations, and the problem perpetuated.

    For this reason we implemented the fixed energy model approach introduced in Section \ref{sec: fixed-model}.  We used the location specific terms we would have used in list $W_{1}$ to create three lists:
    \begin{itemize}
    \item $T_{1}$: A list of character strings or sets of character strings that, if present in a user's profile information, essentially indicate that a user is in the target location.  For example, if a user's profile contains both of the strings ``Casimiro de Abreu'' and ``Brasil,'' we can assume that the user is very likely to be in Casimiro de Abreu, Brazil.
    \item $T_{2}$: A list of character strings that, if present in a user's profile information, strongly suggest that a user is in the target location.  ``Casimiro de Abreu'' is in this list.
    \item $T_{3}$: A list of character strings that, if set as a user's profile location, suggest that a user could be in the target location.  ``Brasil'' is in this list.
    \end{itemize}

    These lists are enumerated in Table \ref{table: Casimiro_lists}.    We used these lists to categorize users according to the following algorithm:
    \begin{enumerate}
    \item If a user's profile meets any of the criteria in list $T_{1}$, assign category A;
    \item Else if a user's profile location contains a string from the world cities list in Appendix \ref{appendix: world_cities}, excepting Rio de Janeiro, assign category B;
    \item Else if a user's profile location, description, name, or screen name contains a string from list $T_{2}$, assign category C;
    \item Else if a user's profile location \emph{is equal to} a string from list $T_{3}$, or if the profile location is empty, assign category D;
    \item Else assign category E.
    \end{enumerate}

    \begin{table}[!htb]
    \centering
    \fontsize{11}{18}\selectfont
    \caption{Categorization lists for Casimiro de Abreu, Brazil.}\label{table: Casimiro_lists}
    \vspace{6pt}
    \begin{tabular}{l|l|l}
    \multicolumn{1}{c|}{$T_{1}$} & \multicolumn{1}{c|}{$T_{2}$} & \multicolumn{1}{c}{$T_{3}$} \\ \hline
    ``Casimiro de Abreu'' AND ``Brasil'' & ``Casimiro de Abreu'' & ``Brazil'' \\
    ``Casimiro de Abreu, RJ'' & & ``Barra de S\~ao Jo\~ao''  \\
    ``Casimiro de Abreu'' AND ``Rio de Janeiro'' & & ``Brasil'' \\
    & & ``Brazil''
    \end{tabular}
    \end{table}

    We then applied the odds table given in Table \ref{table: implementation_odds} to construct the profile energy function.  These odds can be thought of as relationship thresholds required to classify a user in each category into the target location.  Note that users in category C, whose profiles contain keywords or phrases from list $T_{2}$, are assumed to have a higher likelihood of being outside of the target location.  However, a relatively weak relationship with a user inside the target location would be enough to overcome these odds.  On the other hand, a user in category B would require many strong connections with users in the target location in order to be classified in the target location.


    \begin{table}[!hbt]
    \fontsize{11}{18}\selectfont
    \centering
    \caption{Odds Table Used in Naive Implementations} \label{table: implementation_odds}
    \vspace{6pt}
    \begin{tabular}{l|cc}
    Category & Odds (location:non-location) & $\phi(\vx_{i},0)-\phi(\vx_{i},1)$ \\ \hline
    A & 50:1 & $\log(50)$ \\
    B & 1:25 & $-\log(25)$ \\
    C & 1:2 & $-\log(2)$ \\
    D & 1:5 & $-\log(5)$ \\
    E & 1:8 & $-\log(8)$ \\
    \end{tabular}
    \end{table}

    \subsubsection{Performance}

    We employed the expand-classify methodology exactly as in the Corinto collection, except that in each iteration we used the fixed odds in Table \ref{table: implementation_odds} as the profile energy model.  After three hours, this implementation collected 99,606 users and had classified 492 of them as being in the target location.  Of the 99,606 users in the dataset, 7,729 of them had geo-located tweets, 53 of which were inside the target radius.  The resulting AUC for this classifier on the geo-located users was 0.89.  Similar to the Corinto results, the classifier achieves approximately 60\% correct detection rate while maintaining a very low false positive rate.

    \subsubsection{Sensitivity Analysis}

    Figure \ref{fig: Casimiro-gamma-sens} shows the sensitivity of the ROC for this classifier on the geo-located users for several values of $\gamma$ and $\lambda$.  The results do not appear to be very sensitive to the value of $\gamma$, but the AUC decreases for lower values of $\lambda$.  Larger values of $\lambda$ also seem to show slightly better performance.  

    \begin{figure}[!htb]
    \centering
    \begin{tabular}{cc}
    \includegraphics[width=0.45\textwidth]{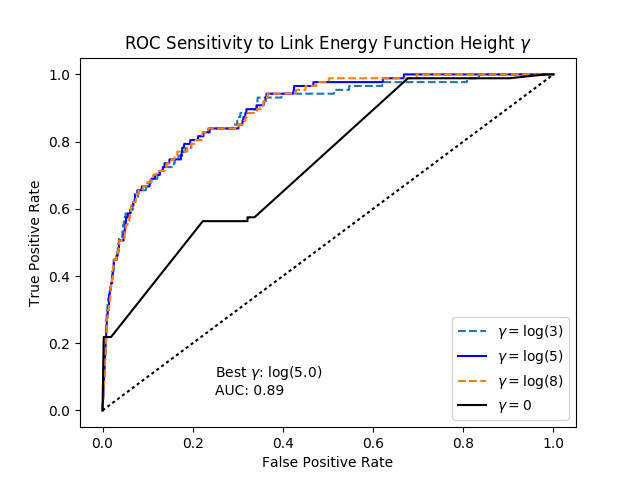}
    &
    \includegraphics[width=0.45\textwidth]{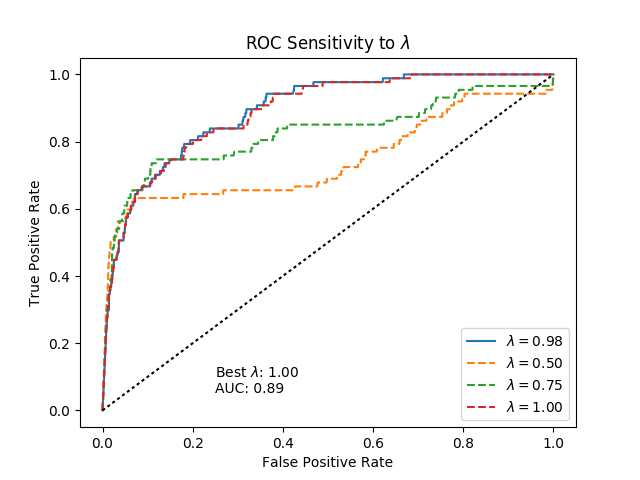}\\
    Sensitivity to $\gamma$
    &
    Sensitivity to $\lambda$
    \end{tabular}
    \caption{Results and Sensitivity of Casimiro de Abreu user classification.} \label{fig: Casimiro-gamma-sens}
    \end{figure}
    Note that the performance of the fixed profile energy function as a classifier is also plotted in this Figure for $\gamma=0$, indicated by a solid black line.  The AUC of this classifier is 0.74. 

    \subsubsection{Summary of Casimiro de Abreu Collection}

    The Casimiro de Abreu collection demonstrates the utility of the expand-classify methodology using a relatively naive approach to forming a profile energy function.  In the case of Casimiro de Abreu, attempts to implement the same approach using the logistic regression classifier were not successful because there were not enough geo-located users in the initial iterations.  Also, we see that our approach produced many more users than searching Twitter.


    \subsection{Caracas, Venezuela}

    Caracas, Venezuela, centered at 10.481\degree N, 66.904\degree W \citep{googlemaps} is the capital of Venezuela. It has a population of approximately 2.1 million, which is much larger than the previous two locations we studied.  We will see this leads to different performance results.    For labeling geo-tagged tweets, we used a radius of 15 miles from the latitude-longitude coordinates above (see Figure \ref{fig: Caracas}).

    \begin{figure}[!htb]
    \centering
    \includegraphics[scale=0.3]{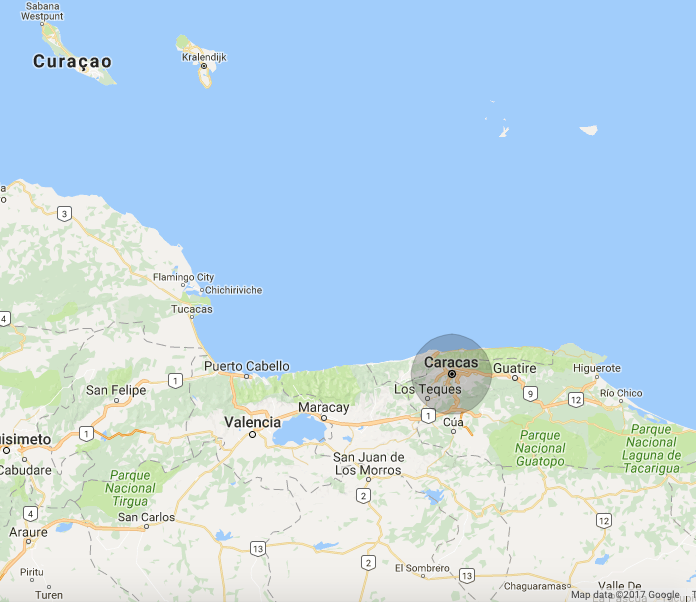}
    \caption{Caracas, Venezuela label radius, plotted on Google Maps\citep{googlemaps}.}\label{fig: Caracas}
    \end{figure}

    \subsubsection{Seed Set}

    Unlike the previous two locations, querying the Twitter API for Caracas, Venezuela returned 983 user profiles, which is close to the API-imposed maximum of 1000.  Of these we used a set of 64 profiles as seed accounts for this collection.

    \subsubsection{Performance}

    We ran the expand-classify algorithm, using logistic regression as the profile energy model, for six hours to collect users in Caracas.    The set of character strings $W_{1}$ used to extract features from the user profiles is given in Table \ref{table: Caracas_terms}.  For $W_{2}$ we again used the list of cities in Appendix \ref{appendix: world_cities}, with Caracas removed.  At the end of the six hour period we had 210,656 users in our dataset, 33,261 of which were classified as being in Caracas.

    \begin{table}[!hbt]
    \fontsize{11}{18}\selectfont
    \centering \fontsize{11}{18}\selectfont
    \caption{List of character strings $W_{1}$ used to extract profile features for Caracas logistic regression.} \label{table: Caracas_terms}
    \vspace{6pt}
    \begin{tabular}{llll}
            ``Caracas'' &
            ``Caracas, Venezuela'' &
            ``Capital, Venezuela'' &
            ``caracas, vzla'' \\
            ``capital, vzla'' &
            ``Caracas Venezuela'' &
            ``Capital Venezuela'' &
            ``caracas vzla'' \\
            ``capital vzla'' &
            ``distrito capital'' &
            ``venezuela'' &
            ``vzla''
      \end{tabular}
    \end{table}
    In this case our method did not produce the level of performance achieved on the smaller locations.  In fact, the best AUC (0.78) was achieved by simply applying the logistic regression classifier on the profiles; including relationship information does not result in significantly improved performance.  Figure \ref{fig: Caracas-gamma-sens} shows the $\gamma$ sensitivity plot.

    \begin{figure}[!htb]
    \centering
    \includegraphics[scale=0.6]{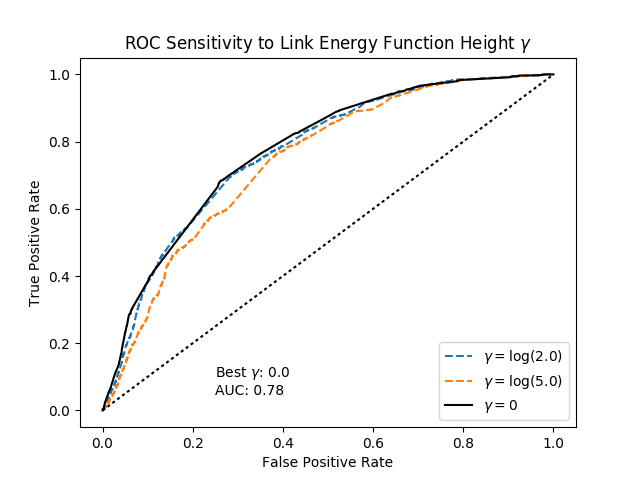}
    \caption{Caracas, Venezuela Performance.}\label{fig: Caracas-gamma-sens}
    \end{figure}
    
    \subsubsection{Discussion of Caracas Collection}

    An obvious difference between Caracas and the other locations presented is its larger population, and these results give us some indication of the limitations of our approach.  Some of the challenges associated with building a set of users from a large city are intuitive: collecting and classifying a larger dataset requires more computational resources.  However, we found a less subtle problem apparent in our attempts to collect users from big cities that relates to our modeling assumptions.  

    \paragraph{Computational Resources.} The first challenge with collecting user sets from large population centers is the problem of computational resources.  Six hours of runtime was not sufficient to collect enough users and links to observe the location homophily implied by our assumptions.  Of the  919 users returned by the Twitter user search API and not used as seed users, only 190 appear in our data. This suggests that we have not run enough \emph{expand} steps to discover many of the users in Caracas.
%
%
%

    \paragraph{Users in Big Cities.} We have assumed that a user whose profile states he or she is in ``Caracas, Venezuela'' is generally going to be in Caracas, but for big cities we have found from our geo-located data that this assumption might not be accurate.      Of the 450 users with geo-located tweets in our dataset, only 205 were located within 15 miles of the center of Caracas.  The remaining 245 were spread throughout the world.  Because of this, a logistic regression classifier  would classify users whose profile locations are ``Caracas, Venezuela'' as being outside of Caracas.  While we observed this property for some user profiles in each location we collected, only in the largest cities did it appear to adversely affect the results.

    Having a preponderance of users whose profiles say they are in Caracas but whose geo-located tweets show they are not brings us to a very important consideration: do we want these users in our target location set?  Some might be Caracas residents who are simply traveling, while others could be studying or working abroad.  Still others might have lived in Caracas in the past but have permanently moved to another location.  Even others could simply be lying.  

    If we do decide these users should be in our dataset, then our approach to fitting a logistic regression model on geo-located users needs to be reworked, because this method of labeling is clearly not a valid proxy for our target set.  If, on the other hand, we do not want these users in the dataset, our method of using the geo-located tweets remains valid. but it brings us to another big city challenge to our assumptions: homophily.

    While we might not know why a user would have a profile location of ``Caracas, Venezuela,'' but tweets geo-located elsewhere in the world, we have observed that many of these users tend to have close connections with other users that appear to be in or near Caracas.  These high-energy, long-distance relationships run counter to our assumption that close relationships tend to indicate shared location.  One plausible conjecture is that people in big cities tend to be more mobile than people from smaller towns, and that mobility has resulted in a larger number of long-distance social media relationships with high link energy scores.  This big-city phenomenon was also observed and documented by \cite{backstrom2010find}.  As a result of it, the homophily that proved useful in the Corinto user classifications is more difficult to exploit in Caracas.

\subsection{Marawi City, Philippines}\label{sec:marawi}
Marawi City, officially known as the Islamic City of Marawi, centered at 8.000\degree N, 124.285\degree E \citep{googlemaps}, is the capital of the province of Lanao del Sur on the island of Mindanao in the Philippines.  It has a population of 201,785, placing it between small cities such as Corinto and large cities such as Caracas in terms of population.  We chose this city for analysis because it has recently been the site of military conflict between the government and members of the extremist group ISIS (Islamic State in Iraq and Syria).  The conflict began on May 23rd, 2017 when the government launched an offensive to capture the leader of an ISIS affiliated group who was reported to be in the city \citep{ref:marawi}.  By May 27th, 2017 nearly 90\% of Marawi's population had been evacuated as the conflict continued.  We will use the content of the users collected from Marawi to detect the onset of the conflict.

\subsubsection{Seed Set}  The seed set of users was obtained through manual search of Twitter.  We obtained 12 users who were clearly in Marawi from the information in their profile.  For instance, some of the users had screen names such as @CITYMARAWI, @choosemarawi, and @marawicity.

\subsubsection{Performance} 
We use logistic regression for the profile energy model.  The set of character strings $W_1$ used are show in Table \ref{table:marawi}.  For $W_2$ we use the list of cities in Appendix \ref{appendix: world_cities}.   After several hours of data collection we had 113,485 users in the dataset, 11,618 of which were classified in Marawi.
As with the other locations studied, we use the geo-located tweets from Marawi as ground-truth to evaluate the classification performance.  
The resulting ROC curve is shown in Figure \ref{fig:marawi_roc}.  At a decision threshold of 0.5, we achieve a true detection rate near 80\% with a false positive rate near 10\%.  The AUC of our approach is 0.89, which indicates a good classification accuracy.   Despite being much larger than the smaller cities we have studied, we are still able to achieve 
a high AUC for Marawi.

    \begin{table}[!hbt]
	\fontsize{11}{18}\selectfont
	\centering
	\caption{List of character strings $W_{1}$ used to extract profile features for Marawi logistic regression.} \label{table:marawi}
	\vspace{6pt}
	\begin{tabular}{llll}
		``Marawi, Philippines''&``Marawi Ph''&``Marawi, Ph''&'`Lanao del Sur''\\
		``Philippines''&``Marawi''&``Mindanao''
	\end{tabular}
\end{table}
We next look at the content posted on Twitter by the nearly 12,000 users our approach classified as being in Marawi.  A simple method to analyze a corpus of text is to use word clouds, which are visualizations in which the size of a word indicates how frequently it occurs in the corpus.  We show the word clouds of the posts of Marawi users for four different dates in Figure \ref{fig:wordcloud}.  As can be seen, for May 21st and 22nd, the most common words are ``BTSBBMAs'' and ``BBMA''.    BBMA stands for the Billboard Music Awards, a music award show which occurred on May 21st, 2017.  ``BTSBBMAs'' refers to the Korean pop band BTS which won an award at the show over other top stars such as Justin Bieber \citep{ref:bts}. Though not shown here, the word clouds for other days before May 23rd show similar pop music related words.  However, the word clouds on May 23rd and 24th show new common phrases, such as ``PrayForMarawi'' and``Allah'' which are associated with the battle in Marawi.  The phrase PrayForMarawi was used to show support for the people of Marawi during the battle.  We found that this phrase was very prominent in the word clouds up to May 25th.  Allah is the name Muslims use for God.  Though Marawi is a Muslim majority city, we did not see Allah in word clouds before May 23rd. Its occurrence after the battle began is most likely tied to people praying for an end to the violence.    

We see here that by analyzing the content of the Marawi users collected with our expand-classify approach, we are able to detect the onset of this political unrest, and with continuing monitoring of these users, we could observe any further developments.  Much of the content was in a mixture of English and local Filipino languages.  Nonetheless, we are still able to quickly assess that an event of significance had occurred.  This shows the utility of being able to monitor a set of social media users from a location.  Even without much prior knowledge one can still easily gain basic situational awareness.
    \begin{figure}[!htb]
	\centering
	\includegraphics[scale=0.5]{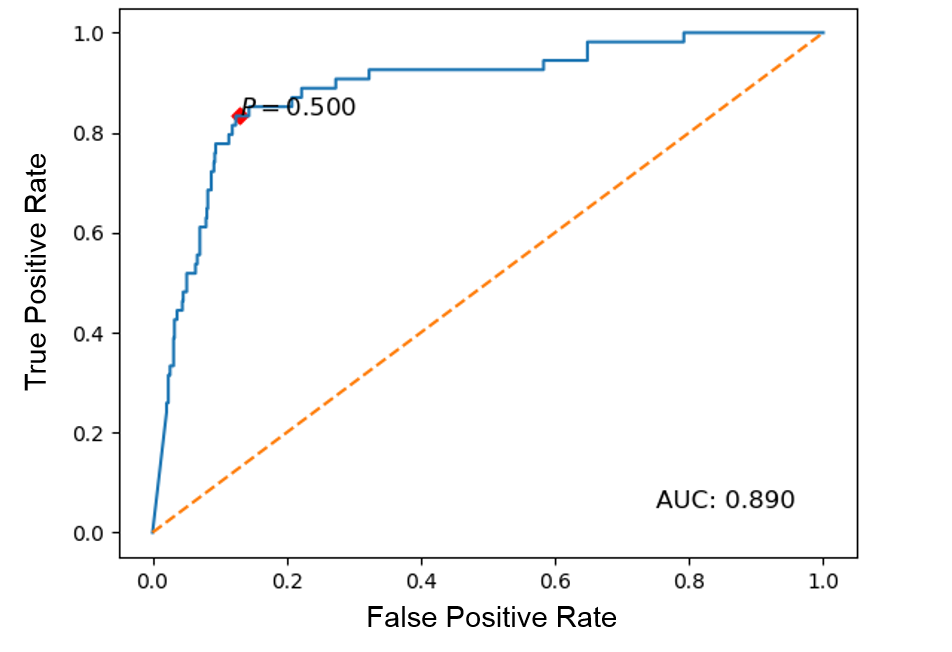}
	\caption{Marawi user classification ROC curve using logistic regression for profile energy function.}\label{fig:marawi_roc}
\end{figure}

    \begin{figure}[!htb]
    \centering
    \includegraphics[scale=0.5]{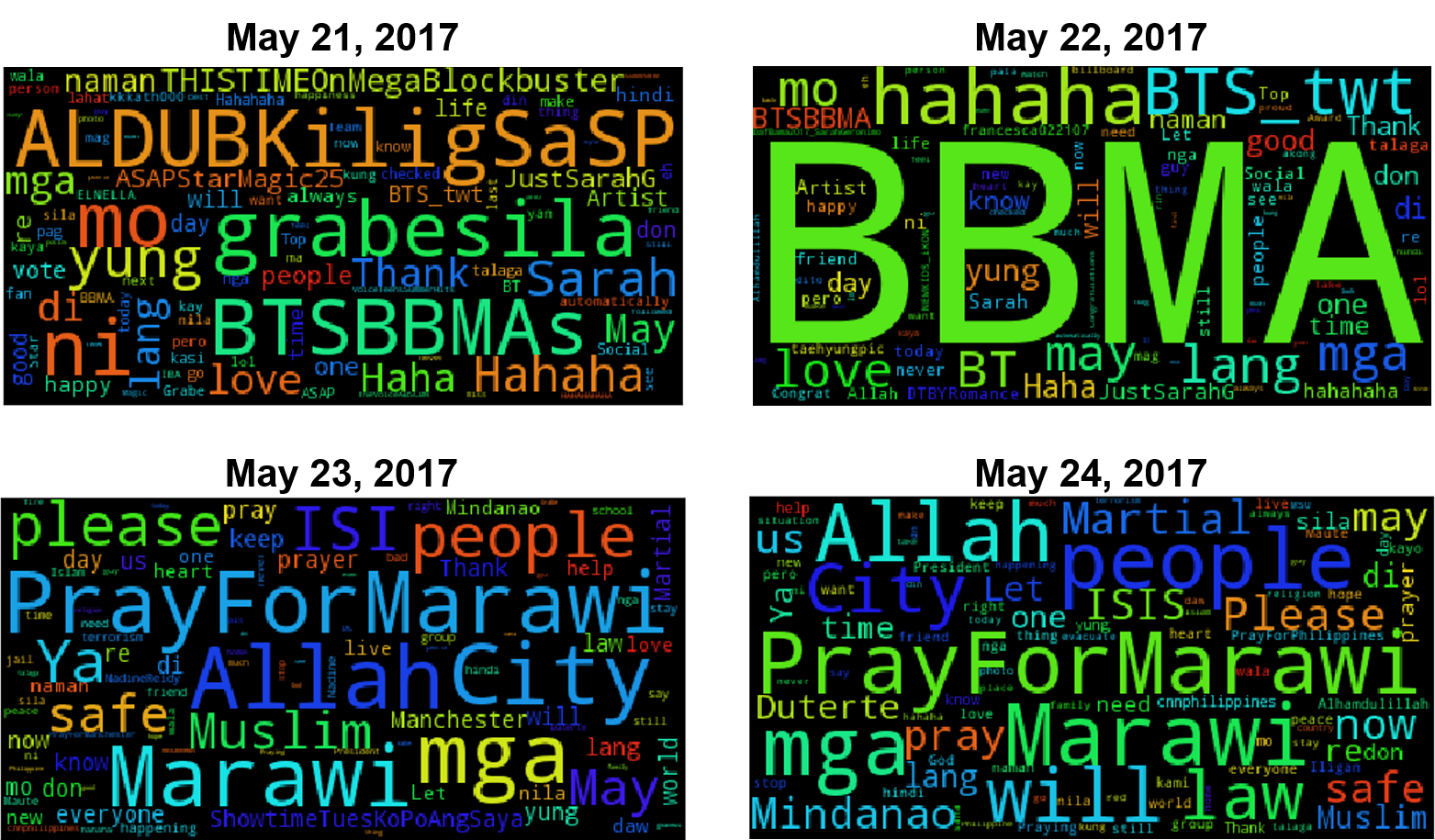}
    \caption{Wordclouds of Marawi users' tweets for different dates.}\label{fig:wordcloud}
    \end{figure}
    \subsection{Summary of Results on All Locations}

    We collected user datasets from multiple locations of varying size and culture.  Table \ref{table: all_results} summarizes the results of these collections, with LR AUC referring to the AUC
		using only logistic regression without any network connection data, and  Model AUC referring to the AUC including network data.  We executed all of the collections for 3-6 hours.  We note that in general, the optimal observed value for $\gamma$ tends to be lower for locations with larger populations.  As we found in our Caracas user collection, there are two likely reasons for this.  First, we did not devote enough computational resources and time to collect and classify the larger set of users.  Second, location-based homophily is less evident in larger population centers, and therefore more difficult to exploit (see \cite{backstrom2010find}).  This is also evident from the larger increase in the AUC for smaller population centers when network data is included.

    Even in small population centers, we find that many of the misclassified geo-located users are those whose profiles and connections indicate they belong in the target location, but whose geo-tagged tweets fall outside of the target location.  This qualitative observation indicates that our method of finding users associated with specific locations is, in some cases, performing better than our evaluation criteria suggest.  

    \begin{table}[!htb]
    \centering
    \fontsize{11}{18}\selectfont
    \caption{Results from user set collections from nine locations, sorted by increasing population.  LR AUC refers to the AUC
		using on logistic regression without any network connection data.  Model AUC is the AUC including network data.}\label{table: all_results}
    \vspace{6pt}
    \begin{tabular}{l|cccc}
    Location & Approx. Population &  LR AUC & Model AUC & Best $\gamma$ \\ \hline
		Zamboanga City, Philippines & 19,542 & 0.81 & 0.87 & $\log (2)$\\
    Corinto, Colombia & 30,000  & 0.75 & 0.92 & $\log (10)$\\
		 Casimiro de Abreu, Brazil & 35,000  & 0.74\raisebox{6pt}{\scriptsize \textdagger} & 0.84 & $\log (8)$ \\
		El Vig\'ia, Venezuela & 156,000  & 0.77 & 0.90  & $\log (6)$ \\
    San Fernando, Venezuela & 165,000  & 0.82 & 0.9 & $\log (8)$ \\
	  Marawi City, Philippines &  201,785 & - & 0.89 & $\log (5)$\\
    Greater Binghamton, NY  & 250,000  & 0.87 & 0.90 & $\log (2)$\\
    Popay\'an, Colombia  & 440,000  & 0.84 & 0.88 & $\log (2)$\\
    Caracas, Venezuela & 2,100,000  & 0.78 & 0.78 & 0 \\
    Asunci\'on, Paraguay & 2,200,000 & 0.64 & 0.7 & $\log (2)$\\
    \multicolumn{5}{l}{\raisebox{6pt}{\scriptsize \textdagger}AUC from fixed profile energy model.}
    \end{tabular}
    \end{table}

\section{Conclusion \& Future Research}

Obtaining a set of social media users from a specific
location is a difficult and important problem. Existing approaches do not allow
one to build large sets of such users or do not accurately retrieve users from the location
of interest.   Our expand-classify approach is able to overcome these limitations and produce
large sets of location-based users for population centers with up to about 500,000 people. 
The expand step allows us to grow a small set of seed users from the location into
very large sets within a few iterations.  The classification step can be efficiently
done by use of a novel factor-graph model which can be optimized using minimum graph-cuts.
Testing based on geo-located content for multiple diverse locations showed that our approach had good accuracy. 
Analysis of content from users from Marawi City obtained with our approach allowed us to detect the onset of political unrest
with minimal prior knowledge of the region.

While the objective of this effort has been to produce a reliable and somewhat comprehensive set of social media users from a specific location, our method could be applied to any social media grouping which exhibits some level of homophily and exhibits some indicators of group membership through users' profile features.  The role of the factor graph model we utilize is to provide a means of considering both the indicators present in a user's profile and the user's social network connections when making classifications.  However, nothing in our model limits these
classification groups to be geographic.   Therefore, our expand-classify approach could be used to build a set of users with a specific political ideology, taste in music, fashion style, or other
features.

Finally, our method might be improved by prioritizing friend and follower queries in the expand step.  Some users might be more inclined to have local friends and followers.  If a probability model can be established that quantifies these inclinations, the network search methods of \cite{alpern2013mining} or \cite{marks2016multi} might be leveraged to grow the set of users in the target location more efficiently.  

\bibliography{loc}

\begin{thebibliography}{36}
\providecommand{\natexlab}[1]{#1}
\providecommand{\url}[1]{\texttt{#1}}
\expandafter\ifx\csname urlstyle\endcsname\relax
  \providecommand{\doi}[1]{doi: #1}\else
  \providecommand{\doi}{doi: \begingroup \urlstyle{rm}\Url}\fi

\bibitem[Airoldi et~al.(2008)Airoldi, Blei, Fienberg, and
  Xing]{airoldi2008mixed}
Edoardo~M Airoldi, David~M Blei, Stephen~E Fienberg, and Eric~P Xing.
\newblock Mixed membership stochastic blockmodels.
\newblock \emph{Journal of Machine Learning Research}, 9\penalty0
  (Sep):\penalty0 1981--2014, 2008.

\bibitem[Alpern and Lidbetter(2013)]{alpern2013mining}
Steve Alpern and Thomas Lidbetter.
\newblock Mining coal or finding terrorists: The expanding search paradigm.
\newblock \emph{Operations Research}, 61\penalty0 (2):\penalty0 265--279, 2013.

\bibitem[Backstrom et~al.(2010)Backstrom, Sun, and Marlow]{backstrom2010find}
Lars Backstrom, Eric Sun, and Cameron Marlow.
\newblock Find me if you can: improving geographical prediction with social and
  spatial proximity.
\newblock In \emph{Proceedings of the 19th international conference on World
  wide web}, pages 61--70. ACM, 2010.

\bibitem[Bo et~al.(2012)Bo, Cook, and Baldwin]{bo2012geolocation}
Han Bo, Paul Cook, and Timothy Baldwin.
\newblock Geolocation prediction in social media data by finding location
  indicative words.
\newblock In \emph{Proceedings of COLING}, pages 1045--1062, 2012.

\bibitem[Boykov et~al.(2001)Boykov, Veksler, and Zabih]{boykov2001fast}
Yuri Boykov, Olga Veksler, and Ramin Zabih.
\newblock Fast approximate energy minimization via graph cuts.
\newblock \emph{IEEE Transactions on pattern analysis and machine
  intelligence}, 23\penalty0 (11):\penalty0 1222--1239, 2001.

\bibitem[Compton et~al.(2014)Compton, Jurgens, and
  Allen]{compton2014geotagging}
Ryan Compton, David Jurgens, and David Allen.
\newblock Geotagging one hundred million twitter accounts with total variation
  minimization.
\newblock In \emph{Big Data (Big Data), 2014 IEEE International Conference on},
  pages 393--401. IEEE, 2014.

\bibitem[Davis(2012)]{davis2012unsupervised}
George~B Davis.
\newblock Unsupervised spatial, temporal and relational models for social
  processes.
\newblock Technical report, DTIC Document, 2012.

\bibitem[Davis~Jr. et~al.(2011)Davis~Jr., Pappa, de~Oliveira, and
  de~L.~Arcanjo]{Davis-InferringLocation}
Clodoveu~A. Davis~Jr., Gisele~L. Pappa, Diogo Renn\'o~Rocha de~Oliveira, and
  Filipe de~L.~Arcanjo.
\newblock Inferring the location of twitter messages based on user
  relationships.
\newblock \emph{Transactions in GIS}, 15\penalty0 (6):\penalty0 735--751, 2011.
\newblock ISSN 1467-9671.
\newblock \doi{10.1111/j.1467-9671.2011.01297.x}.
\newblock URL \url{http://dx.doi.org/10.1111/j.1467-9671.2011.01297.x}.

\bibitem[Gao et~al.(2011)Gao, Barbier, and Goolsby]{gao2011harnessing}
Huiji Gao, Geoffrey Barbier, and Rebecca Goolsby.
\newblock Harnessing the crowdsourcing power of social media for disaster
  relief.
\newblock \emph{IEEE Intelligent Systems}, 26\penalty0 (3):\penalty0 10--14,
  2011.

\bibitem[Gilbert and Karahalios(2009)]{gilbert2009predicting}
Eric Gilbert and Karrie Karahalios.
\newblock Predicting tie strength with social media.
\newblock In \emph{Proceedings of the SIGCHI conference on human factors in
  computing systems}, pages 211--220. ACM, 2009.

\bibitem[google.com(2017)]{googlemaps}
google.com.
\newblock Google {M}aps, 2017.
\newblock URL \url{maps.google.com}.

\bibitem[Han et~al.(2013)Han, Cook, and Baldwin]{han2013stacking}
Bo~Han, Paul Cook, and Timothy Baldwin.
\newblock A stacking-based approach to twitter user geolocation prediction.
\newblock In \emph{ACL (Conference System Demonstrations)}, pages 7--12, 2013.

\bibitem[Hecht et~al.(2011)Hecht, Hong, Suh, and Chi]{hecht2011tweets}
Brent Hecht, Lichan Hong, Bongwon Suh, and Ed~H Chi.
\newblock Tweets from justin bieber's heart: the dynamics of the location field
  in user profiles.
\newblock In \emph{Proceedings of the SIGCHI conference on human factors in
  computing systems}, pages 237--246. ACM, 2011.

\bibitem[Jurgens(2013)]{jurgens2013s}
David Jurgens.
\newblock That's what friends are for: Inferring location in online social
  media platforms based on social relationships.
\newblock \emph{ICWSM}, 13:\penalty0 273--282, 2013.

\bibitem[Kolmogorov and Zabin(2004)]{kolmogorov2004energy}
Vladimir Kolmogorov and Ramin Zabin.
\newblock What energy functions can be minimized via graph cuts?
\newblock \emph{IEEE transactions on pattern analysis and machine
  intelligence}, 26\penalty0 (2):\penalty0 147--159, 2004.

\bibitem[Kong et~al.(2014)Kong, Liu, and Huang]{kong2014spot}
Longbo Kong, Zhi Liu, and Yan Huang.
\newblock Spot: Locating social media users based on social network context.
\newblock \emph{Proceedings of the VLDB Endowment}, 7\penalty0 (13):\penalty0
  1681--1684, 2014.

\bibitem[Kumar et~al.(2013)Kumar, Morstatter, Zafarani, and Liu]{kumar2013whom}
Shamanth Kumar, Fred Morstatter, Reza Zafarani, and Huan Liu.
\newblock Whom should i follow?: identifying relevant users during crises.
\newblock In \emph{Proceedings of the 24th ACM conference on Hypertext and
  social media}, pages 139--147. ACM, 2013.

\bibitem[Leskovec and Mcauley(2012)]{NIPS2012-4532}
Jure Leskovec and Julian~J. Mcauley.
\newblock Learning to discover social circles in ego networks.
\newblock In F.~Pereira, C.J.C. Burges, L.~Bottou, and K.Q. Weinberger,
  editors, \emph{Advances in Neural Information Processing Systems 25}, pages
  539--547. Curran Associates, Inc., 2012.
\newblock URL
  \url{http://papers.nips.cc/paper/4532-learning-to-discover-social-circles-in-ego-networks.pdf}.

\bibitem[Li et~al.(2012)Li, Wang, and Chang]{li2012multiple}
Rui Li, Shengjie Wang, and Kevin Chen-Chuan Chang.
\newblock Multiple location profiling for users and relationships from social
  network and content.
\newblock \emph{Proceedings of the VLDB Endowment}, 5\penalty0 (11):\penalty0
  1603--1614, 2012.

\bibitem[Liu(2017)]{ref:bts}
Marian Liu.
\newblock {Bigger than Bieber? K-pop group BTS beats US stars to win Billboard
  Music Award}.
\newblock \emph{CNN}, May 22, 2017.
\newblock URL
  \url{http://www.cnn.com/2017/05/22/entertainment/k-pop-bts-billboard-music-awards/index.html}.

\bibitem[Marks and Zaman(2016)]{marks2016multi}
Christopher~E Marks and Tauhid Zaman.
\newblock A multi-urn model for network search.
\newblock \emph{arXiv preprint arXiv:1608.08080}, 2016.

\bibitem[Matsuzawa(2017)]{ref:marawi}
Mikas Matsuzawa.
\newblock {Marawi crisis: What we know—and don't know—so far}.
\newblock \emph{The Philippine Star}, May 29, 2017.
\newblock URL
  \url{http://www.philstar.com/headlines/2017/05/29/1703153/marawi-crisis-what-we-know-and-dont-know-so-far}.

\bibitem[MaxMind.com(2017)]{maxmind}
MaxMind.com.
\newblock World cities database, 2017.
\newblock URL \url{http://www.maxmind.com/}.

\bibitem[McGee et~al.(2013)McGee, Caverlee, and Cheng]{mcgee2013location}
Jeffrey McGee, James Caverlee, and Zhiyuan Cheng.
\newblock Location prediction in social media based on tie strength.
\newblock In \emph{Proceedings of the 22nd ACM international conference on
  Information \& Knowledge Management}, pages 459--468. ACM, 2013.

\bibitem[Merchant et~al.(2011)Merchant, Elmer, and
  Lurie]{merchant2011integrating}
Raina~M Merchant, Stacy Elmer, and Nicole Lurie.
\newblock Integrating social media into emergency-preparedness efforts.
\newblock \emph{New England Journal of Medicine}, 365\penalty0 (4):\penalty0
  289--291, 2011.

\bibitem[Pedregosa et~al.(2011)Pedregosa, Varoquaux, Gramfort, Michel, Thirion,
  Grisel, Blondel, Prettenhofer, Weiss, Dubourg, Vanderplas, Passos,
  Cournapeau, Brucher, Perrot, and Duchesnay]{scikit-learn}
F.~Pedregosa, G.~Varoquaux, A.~Gramfort, V.~Michel, B.~Thirion, O.~Grisel,
  M.~Blondel, P.~Prettenhofer, R.~Weiss, V.~Dubourg, J.~Vanderplas, A.~Passos,
  D.~Cournapeau, M.~Brucher, M.~Perrot, and E.~Duchesnay.
\newblock Scikit-learn: Machine learning in {P}ython.
\newblock \emph{Journal of Machine Learning Research}, 12:\penalty0 2825--2830,
  2011.

\bibitem[Rout et~al.(2013)Rout, Bontcheva, Preo{\c{t}}iuc-Pietro, and
  Cohn]{rout2013s}
Dominic Rout, Kalina Bontcheva, Daniel Preo{\c{t}}iuc-Pietro, and Trevor Cohn.
\newblock Where's @wally?: a classification approach to geolocating users based
  on their social ties.
\newblock In \emph{Proceedings of the 24th ACM Conference on Hypertext and
  Social Media}, pages 11--20. ACM, 2013.

\bibitem[Starbird et~al.(2012)Starbird, Muzny, and Palen]{starbird2012learning}
Kate Starbird, Grace Muzny, and Leysia Palen.
\newblock Learning from the crowd: Collaborative filtering techniques for
  identifying on-the-ground twitterers during mass disruptions.
\newblock In \emph{Proc. 9th Int. Conf. Inf. Syst. Crisis Response Manag.
  Iscram}, 2012.

\bibitem[Sutton et~al.(2008)Sutton, Palen, and
  Shklovski]{sutton2008backchannels}
Jeannette~N Sutton, Leysia Palen, and Irina Shklovski.
\newblock \emph{Backchannels on the front lines: Emergency uses of social media
  in the 2007 Southern California Wildfires}.
\newblock University of Colorado, 2008.

\bibitem[Twitter(2016{\natexlab{a}})]{rest}
Twitter.
\newblock Twitter {REST API}, 2016{\natexlab{a}}.
\newblock \url{https://dev.twitter.com/rest/public}.

\bibitem[Twitter(2016{\natexlab{b}})]{rest-search}
Twitter.
\newblock Twitter {GET} search/tweets endpoint, {REST API}, 2016{\natexlab{b}}.
\newblock \url{https://dev.twitter.com/rest/reference/get/search/tweets}.

\bibitem[Twitter(2016{\natexlab{c}})]{rest-usersearch}
Twitter.
\newblock Twitter {GET} users/search endpoint, {REST API}, 2016{\natexlab{c}}.
\newblock \url{https://dev.twitter.com/rest/reference/get/users/search}.

\bibitem[Yates and Paquette(2011)]{yates2011emergency}
Dave Yates and Scott Paquette.
\newblock Emergency knowledge management and social media technologies: A case
  study of the 2010 haitian earthquake.
\newblock \emph{International journal of information management}, 31\penalty0
  (1):\penalty0 6--13, 2011.

\bibitem[{Yesica Fisch, Associated Press}(2017)]{usnews-casimiro}
{Yesica Fisch, Associated Press}.
\newblock Small brazil city on edge after man dies from yellow fever.
\newblock \emph{{U.S. N}ews}, March 18, 2017.
\newblock URL
  \url{https://www.usnews.com/news/news/articles/2017-03-18/small-brazil-city-on-edge-after-man-dies-from-yellow-fever}.

\bibitem[Yin et~al.(2012)Yin, Lampert, Cameron, Robinson, and
  Power]{yin2012using}
Jie Yin, Andrew Lampert, Mark Cameron, Bella Robinson, and Robert Power.
\newblock Using social media to enhance emergency situation awareness.
\newblock \emph{IEEE Intelligent Systems}, 27\penalty0 (6):\penalty0 52--59,
  2012.

\bibitem[Zabih and Kolmogorov(2004)]{zabih2004spatially}
Ramin Zabih and Vladimir Kolmogorov.
\newblock Spatially coherent clustering using graph cuts.
\newblock In \emph{Computer Vision and Pattern Recognition, 2004. CVPR 2004.
  Proceedings of the 2004 IEEE Computer Society Conference on}, volume~2, pages
  II--437. IEEE, 2004.

\end{thebibliography}
\bibliographystyle{plainnat}

\newpage
\APPENDIX{}
\section{Proof of Theorem \ref{thm: optimality}}\label{sec:proof_thm}
 Before proving Theorem \ref{thm: optimality}, it is useful to establish our notation and provide an important Lemma.  We let $\mathcal{G}=(\mathcal{V},\mathcal{E})$ be the Energy Graph representation of for an energy function of the form in equation \eqref{eq: energy_function}, which adheres to the assumptions in Section \ref{sec: assumptions}, where $\mathcal{V}$ is the set of nodes and $\mathcal{E}$ is the set of edges.  A valid $s$-$t$ cut is a partition of $\mathcal{V}$ into two subsets: $S$ and $T$, where source node $s \in S$ and sink node $t \in T$.  We define the cut-set $C\subset \mathcal{E}$ as the set of directed edges going from any node in set $S$ to any node in set $T$.  
For arbitrary location classifications $\mathbf{L} = (\ell_{1},\ell_{2},\ldots,\ell_{N}) \in \{0,1\}^{N}$,
consider the $s$-$t$ cut on the Energy Graph $\mathcal{G}$ that partitions the nodes according to their location classes.  We denote the set of nodes in the subset belonging to the source node as $S_{\mathbf{L}}$, where node $u_{i} \in S_{\mathbf{L}}$ for all users $i$ for which $\ell_{i}=1$.  Likewise, $u_{j} \in T_{\mathbf{L}}$ for all $j$ such that $\ell_{j}=0$.  We refer to this cut as the $\mathbf{L}$-\emph{configuration} cut on the Energy Graph.

We state our Lemma.
\begin{lemma}[Graph Equivalence] \label{lem}
	Suppose we are given an energy function of the form
	\[
	E(\mathbf{L})=\sum_{i}\phi(\vx_{i},\ell_{i}) + \sum_{i<j}\psi(\vz_{i,j},\ell_{i},\ell_{j}),
	\]
	where functions $\phi(\vx_{i},\ell_{i})$ and $\psi(\vz_{i,j},\ell_{i},\ell_{j})$ adhere to the assumptions in Section \ref{sec: assumptions}.  Then, for arbitrary location classification vector $\mathbf{L} \in \{0,1\}^{N}$, the value of the function $E(\mathbf{L})$ is equal to the capacity of the  $\mathbf{L}$-configuration cut on the corresponding energy graph.
\end{lemma}

\begin{proof}{Proof of Lemma \ref{lem}.}
	Let $G=(\mathcal{V},\mathcal{E})$ be the energy graph representation of energy function $E$, and consider an arbitrary fixed location classification vector $\mathbf{L} \in \{0,1\}^{N}$.  From our definition of the $\mathbf{L}$-configuration cut, $\ell_{i}=1$ implies node $u_{i} \in S_{\mathbf{L}}$ and $(u_{i},t) \in C_{\mathbf{L}}$.  Likewise, $\ell_{i}=0$ implies $(S,u_{i}) \in C_{\mathbf{L}}$.  
	Now consider an arbitrary pair of user nodes $u_{i},u_{j}, \ i \neq j$.  One of the following cases apply: (Case 1) both users are assigned to location class $1$ ($\ell_{i}=\ell_{j}=1$), (Case 2) one of the users is in location class 1 and the other is assigned to location class 0 ($\ell_{i} \neq \ell_{j}$), or (Case 3) both users are assigned to location class 0 ($\ell_{i}=\ell_{j}=0$).  We consider the implications of each case on cut-set $C_{\mathbf{L}}$.
	
	{\bf Case 1.}
	Because $u_{i}$ and $u_{j}$ are both in set $S_{\mathbf{L}}$, edges between these two nodes are not in $C_{\mathbf{L}}$.  It follows that 
	\[
	\ell_{i}=\ell_{j}=1 \Rightarrow 
	(u_{i},u_{j}),(u_{j},u_{i}) \notin C_{\mathbf{L}}.
	\]
	However, edges $(u_{i},t)$ and $(u_{j},t)$ are in the cut-set $C_{\mathbf{L}}$.
	Figure \ref{fig: case1} provides an illustration of this case.  
	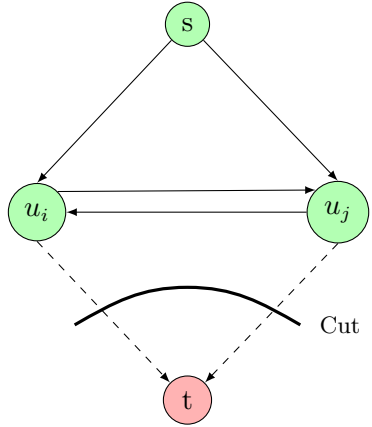
\begin{figure}
		\centering
		\begin{tikzpicture}
		\node[circle,draw,fill=green!30!white](u1) at (0,0){$u_{i}$};
		\node[circle,draw,fill=green!30!white](u2) at (4,0){$u_{j}$};
		\node[circle,draw,fill=red!30!white](t) at (2,-2.5){t};
		\node[circle,draw,fill=green!30!white](s) at (2,2.5){s};
		\path[draw,-latex](s.south west) -- (u1.north);
		\path[draw,-latex](s.south east) -- (u2.north);
		\path[draw,-latex](u1.north east) -- (u2.north west);
		\path[draw,-latex](u2.west) -- (u1.east);
		\path[draw,-latex,dashed](u1.south)--(t.north west);
		\path[draw,-latex,dashed](u2.south)--(t.north east);
		\draw[very thick] 
		(0.5,-1.5) 
		to[out=30,in=180] 
		(2,-1) 
		to[out=0,in=150] 
		(3.5,-1.5) node[above,right]{\footnotesize Cut};
		\end{tikzpicture}
		\caption{Illustration of Case 1 cut.  Nodes in set $S_{\mathbf{L}}$ are shaded green, while nodes in set $T_{\mathbf{L}}$ are shaded red.  Dashed edges are in cut-set $C_{\mathbf{L}}$.} \label{fig: case1}
	\end{figure}
	
	{\bf Case 2.}
	Without loss of generality, assume $\ell_{i}=1$ and $\ell_{j}=0$, implying $(u_{i},t) \in C_{\mathbf{L}}$ and $(s,u_{j}) \in C_{\mathbf{L}}$.  Because $u_{i} \in S$ and $u_{j}\in T$, edge $(u_{i},u_{j})$ is also in the cut-set $C_{\mathbf{L}}$.  This case is depicted in Figure \ref{fig: case2}.
	
	Note that the reverse edge from $u_{j}$ to $u_{i}$ in Figure \ref{fig: case2} is \emph{not} in the cut-set because they go from set $T_{\mathbf{L}}$ to set $S_{\mathbf{L}}$. 
	
	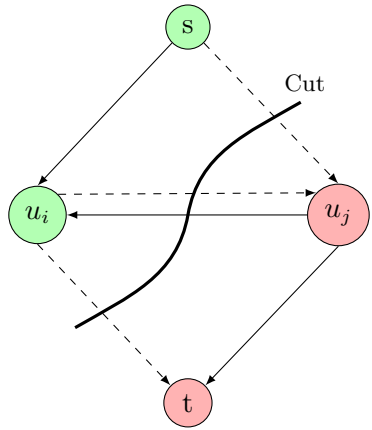
\begin{figure}[!hbt]
		\centering
		\begin{tikzpicture}
		\node[circle,draw,fill=green!30!white](u1) at (0,0){$u_{i}$};
		\node[circle,draw,fill=red!30!white](u2) at (4,0){$u_{j}$};
		\node[circle,draw,fill=red!30!white](t) at (2,-2.5){t};
		\node[circle,draw,fill=green!30!white](s) at (2,2.5){s};
		\path[draw,-latex](s.south west) -- (u1.north);
		\path[draw,-latex,dashed](s.south east) -- (u2.north);
		\path[draw,-latex,dashed](u1.north east) -- (u2.north west);
		\path[draw,-latex](u2.west) -- (u1.east);
		\path[draw,-latex,dashed](u1.south)--(t.north west);
		\path[draw,-latex](u2.south)--(t.north east);
		\draw[very thick] 
		(0.5,-1.5) 
		to[out=30,in=-100] 
		(2,0) node[left]{} 
		to[out=80,in=-150]
		(3.5,1.5) 
		node[above] {\footnotesize Cut};
		\end{tikzpicture}
		\caption{Illustrations of both minimum cut possibilities for Case 2.   Nodes in set $S_{\mathbf{L}}$ are shaded green, while nodes in set $T_{\mathbf{L}}$ are shaded red.  Dashed edges are in cut-set $C_{\mathbf{L}}$.} \label{fig: case2}
	\end{figure}
	
	{\bf Case 3.}  Finally, we consider the case in which $\ell_{i}=\ell_{j}=0$.  This case is very similar to Case 1: edges $(s,u_{i})$ and $(s,u_{j})$ are in $C_{\mathbf{L}}$, while other edges incident to these nodes are not in the cut-set (see Figure \ref{fig: case3}).
	

	\begin{figure}[!hbt]
		\centering
		\begin{tikzpicture}
		\node[circle,draw,fill=red!30!white](u1) at (0,0){$u_{i}$};
		\node[circle,draw,fill=red!30!white](u2) at (4,0){$u_{j}$};
		\node[circle,draw,fill=red!30!white](t) at (2,-2.5){t};
		\node[circle,draw,fill=green!30!white](s) at (2,2.5){s};
		\path[draw,-latex,dashed](s.south west) -- (u1.north);
		\path[draw,-latex,dashed](s.south east) -- (u2.north);
		\path[draw,-latex](u1.north east) -- (u2.north west);
		\path[draw,-latex](u2.west) -- (u1.east);
		\path[draw,-latex](u1.south)--(t.north west);
		\path[draw,-latex](u2.south)--(t.north east);
		\draw[very thick] 
		(0.5,1.5) 
		to[out=-30,in=180] 
		(2,1) node[above]{\footnotesize Cut} 
		to[out=0,in=-150] 
		(3.5,1.5);
		\end{tikzpicture}
		\caption{Illustration of Case 3 cut.  Nodes in set $S_{\mathbf{L}}$ are shaded green, while nodes in set $T_{\mathbf{L}}$ are shaded red.  Dashed edges are in cut-set $C_{\mathbf{L}}$.} \label{fig: case3}
	\end{figure}
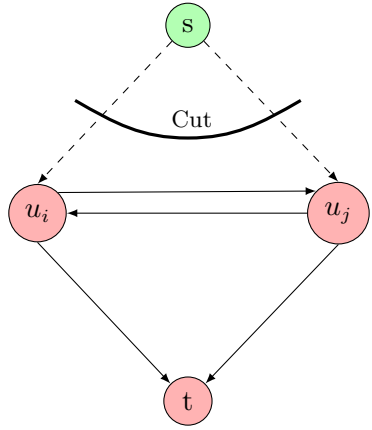

	From these rules we can identify all of the edges comprising set cut-set $C_{\mathbf{L}}$ for an arbitrary location class vector $\mathbf{L}$. The total capacity of the cut is the sum of all of these edge capacities:
	\begin{align*}
	\sum_{a \in C_{\mathbf{L}}} c_{a} & = \underbrace{\sum_{i: \; \ell_{i}=0} c_{(s,u_{i})} + \sum_{i: \; \ell_{i}=1} c_{(u_{i},t)}}_{\mathbf{L}\mathrm{-configuration}} 
	+ \underbrace{\sum_{i<j:\; \ell_{i}=1,\; \ell_{j}=0} c_{(u_{j},u_{i})}
		+ \sum_{i<j:\; \ell_{i}=0,\; \ell_{j}=1} c_{(u_{i},u_{j})}}_{\mathrm{Case \ 2 \ user \ node \ links}} \\
	& = \sum_{i: \; \ell_{i}=0} \left(\phi(\vx_{i},0) +\frac{1}{2}\sum_{j:\; \neq i} \psi(\vz_{i,j},0,0) \right)
	+\sum_{i:\;\ell_{i}=1} \phi(\vx_{i},1) 
	\\ & \quad 
	+ \sum_{i<j:\; \ell_{i}=1,\; \ell_{j}=0} \left(\psi(\vz_{i,j},1,0) - \frac{1}{2}\psi(\vz_{i,j},0,0)\right)
	+ \sum_{  i<j:\; \ell_{i}=0,\; \ell_{j}=1} \left(\psi(\vz_{i,j},0,1) - \frac{1}{2}\psi(\vz_{i,j},0,0)\right) \\
	& = \sum_{i=1}^{N} \phi(\vx_{i},\ell_{i})
	+ \sum_{i<j:\;\ell_{i}=\ell_{j}=0} \psi(\vz_{i,j},0,0)
	+\sum_{i<j:\; \ell_{i}=1,\; \ell_{j}=0} 
	\psi(\vz_{i,j},1,0) 
	+ \sum_{  i<j:\; \ell_{i}=0,\; \ell_{j}=1} 
	\psi(\vz_{i,j},0,1) 
	\\
	& = \sum_{i=1}^{N} \phi(\vx_{i},\ell_{i}) + \sum_{i < j} \psi(\vz_{i,j},\ell_{i},\ell_{j}) = E(\mathbf{L}). \halmos
	\end{align*}
\end{proof}

\begin{proof}{Proof of Theorem \ref{thm: optimality}.}
	Theorem \ref{thm: optimality} follows almost immediately from Lemma \ref{lem}.  Suppose we find the minimum $s$-$t$ cut on the Energy Graph corresponding to energy function $E(\mathbf{L})$, and let sets $S$ and $T$ be the corresponding partition of $\mathcal{V}$ and $C \subset \mathcal{E}$ be the cut-set.  In order for the cut to be valid, each user node $u_{i}$ must be in either set $S$ or set $T$.  For each user node $u_{i} \in S$, edge $(u_{i},t)$ is in the cut-set $C$ and we set $\ell_{i}=1$.  For each user node $u_{i} \in T$, edge $(s,u_{i}) \in C$ and we set $\ell_{i}=0$.  Let $\mathbf{L}^{\star}$ be the resulting vector of location assignments.  
	
	From Lemma \ref{lem}, the capacity of this cut is $E(\mathbf{L}^{\star})$.  Because the sets $S$ and $T$ were constructed from the minimum capacity $s$-$t$ cut on the graph, there cannot be another location assignment $\mathbf{L}$ for which $E(\mathbf{L}) < E(\mathbf{L}^{\star})$, as such a vector would allow for the construction of an $s$-$t$ cut with lower capacity.  \halmos
\end{proof}

\section{World Cities Data} \label{appendix: world_cities}

The following cities were used in set $W_{2}$ in our specific location collection implementations.  This list is extracted from the World Cities Dataset created and maintained by MaxMind, available at 
http://www.maxmind.com/ \citep{maxmind}\\[12pt]

\begin{center} \fontsize{10pt}{10pt}\selectfont
\begin{tabular}{llllll}
dubai & kabul & yerevan & luanda & cordoba & rosario \\ 
vienna & adelaide & brisbane & melbourne & perth & sydney \\ 
baku & dhaka & khulna & brussels & ouagadougou & sofia \\ 
belem & belo horizonte & brasilia & campinas & curitiba & fortaleza \\ 
goiania & guarulhos & manaus & nova iguacu & porto alegre & recife \\ 
rio de janeiro & salvador & sao paulo & minsk & montreal & toronto \\ 
vancouver & kinshasa & lubumbashi & brazzaville & abidjan & santiago \\ 
douala & yaounde & anshan & changchun & chengdu & chongqing \\ 
dalian & datong & fushun & fuzhou & guangzhou & guiyang \\ 
handan & hangzhou & harbin & hefei & huainan & jilin \\ 
jinan & kunming & lanzhou & luoyang & nanchang & nanjing \\ 
peking & qingdao & rongcheng & shanghai & shenyang & shenzhen \\ 
suzhou & taiyuan & tangshan & tianjin & urumqi & wuhan \\ 
wuxi & xian & xianyang & xinyang & xuzhou & barranquilla \\ 
bogota & cali & medellin & prague & berlin & hamburg \\ 
munich & copenhagen & santo domingo & algiers & guayaquil & quito \\ 
alexandria & cairo & gizeh & barcelona & madrid & addis abeba \\ 
paris & london & tbilisi & accra & kumasi & conakry \\ 
port-au-prince & budapest & bandung & bekasi & depok & jakarta \\ 
makasar & medan & palembang & semarang & surabaya & tangerang \\ 
dublin & agra & ahmadabad & allahabad & amritsar & aurangabad \\ 
bangalore & bhopal & bombay & calcutta & delhi & faridabad \\ 
ghaziabad & haora & hyderabad & indore & jabalpur & jaipur \\ 
kalyan & kanpur & lakhnau & ludhiana & madras & nagpur \\ 
new delhi & patna & pimpri & pune & rajkot & surat \\ 
thana & vadodara & varanasi & visakhapatnam & baghdad & esfahan \\ 
karaj & mashhad & qom & shiraz & tabriz & milan \\ 
rome & hiroshima & kawasaki & kobe & nagoya & saitama \\ 
tokyo & nairobi & phnum penh & seoul & almaty & bayrut \\ 
beirut & tripoli & casablanca & fez & rabat & antananarivo \\ 
bamako & mandalay & rangoon & ecatepec & guadalajara & juarez \\ 
leon & mexico & monterrey & nezahualcoyotl & puebla & tijuana \\ 
kuala lumpur & maputo & benin & ibadan & kaduna & kano \\ 
lagos & maiduguri & port harcourt & managua & lima & davao \\ 
manila & faisalabad & gujranwala & hyderabad & karachi & lahore \\ 
multan & peshawar & rawalpindi & warsaw & bucharest & belgrade \\ 
chelyabinsk & kazan & moscow & nizhniy novgorod & novosibirsk & omsk \\ 
rostov-na-donu & saint petersburg & samara & ufa & volgograd & yekaterinburg \\ 
jiddah & mecca & riyadh & khartoum & umm durman & stockholm \\ 
singapore & freetown & dakar & mogadishu & aleppo & damascus \\ 
bangkok & adana & ankara & bursa & gaziantep & istanbul \\ 
izmir & kaohsiung & kaohsiung & taichung & taipei & dar es salaam \\ 
kiev & odesa & kampala & phoenix & los angeles & san diego \\ 
chicago & new york & philadelphia & dallas & houston & san antonio \\ 
montevideo & tashkent & caracas & maracaibo & valencia & hanoi \\ 
ha noi & ho chi minh city & cape town & durban & johannesburg & pretoria \\ 
soweto & lusaka & harare &  &  &  \\ 
\end{tabular}
\end{center}
\end{document}